\documentclass[reqno]{amsart}
\usepackage{amsmath,amssymb,color}
\usepackage[mathscr]{euscript}
\usepackage{stmaryrd}
\usepackage{graphicx}
\usepackage[pdftex,breaklinks,colorlinks,
citecolor=blue,
urlcolor=blue,
pdftitle={Hydrodynamic limit of boundary driven exclusion processes with nonreversible boundary dynamics},
pdfauthor={C. Erignoux}]{hyperref}
\usepackage{amsfonts, amscd, epsfig, amsmath, amssymb,enumerate}
\usepackage{graphicx}
\usepackage{color}
\usepackage{mathrsfs}
\usepackage{tikz}
\usetikzlibrary{backgrounds}
\usetikzlibrary{patterns,fadings}
\usetikzlibrary{arrows,decorations.pathmorphing}
\usetikzlibrary{decorations.pathreplacing}
\usetikzlibrary{decorations}
\usetikzlibrary{calc}
\usetikzlibrary{shapes.misc}
\definecolor{light-gray}{gray}{0.95}

\usepackage{float}
\def\centerarc[#1](#2)(#3:#4:#5){\draw[#1] ($(#2)+({#5*cos(#3)},{#5*sin(#3)})$) arc (#3:#4:#5);}

\allowdisplaybreaks 
\makeatletter
\@addtoreset{equation}{section}
\makeatother

\newtheorem{theorem}{Theorem}[section]
\newtheorem{lemma}[theorem]{Lemma}
\newtheorem{proposition}[theorem]{Proposition}
\newtheorem{corollary}[theorem]{Corollary}

\newtheorem{remark}[theorem]{Remark}

\newcommand{\bb}[1]{{\mathbb #1}}
\newcommand{\bs}[1]{{\boldsymbol #1}}

\newcounter{as}[section]

\newcommand{\R}{{\mathbb R}}
\newcommand{\bfk}{{\bs k}}
\newcommand{\bfx}{{\bs x}}

\newcommand{\bfX}{{\bs X}}

\newcommand{\bfZ}{{\bs Z}}
\newcommand{\Z}{{\bb Z}}

\newcommand{\abs}[1]{\left|#1 \right|}

\newcommand{\N}{{\mathbb N}}

\newcommand{\E}{{\mathbb{E}}}
\newcommand{\Prob}{{\mathbb{P}}}
\newcommand{\dd}{{\mathfrak{d}}}
\newcommand{\pa}[1]{\left(#1 \right)}
\newcommand{\0}{{\bf 0}}
\newcommand{\ind}[1]{{\bf 1}_{\{#1\}}}
\newcommand{\norm}[1]{\|#1\|}
\newcommand{\gene}{{\mathcal L}}
\newcommand{\gend}{L^{\dagger}}
\newcommand{\gendb}{\bar L^{\dagger}}
\newcommand{\BFL}{{\boldsymbol \Lambda_N}}
\newcommand{\lps}{\Lambda_p^*}
\newcommand{\cro}[1]{\left[#1\right]}
\newcommand{\br}[1]{\left\{#1\right\}}

\newcommand{\bfN}{{\bs{N}}}

\newcommand{\bfH}{\bs{H}}

\definecolor{dkgreen}{rgb}{0,0.6,0}
\definecolor{gray}{rgb}{0.5,0.5,0.5}
\definecolor{header}{gray}{0.3}

\title[Hydrodynamic limit of the SSEP with nonrev. boundary dynamics]{Hydrodynamic limit of boundary driven exclusion processes with nonreversible boundary dynamics}
\author{C. Erignoux}
\email{clement.erignoux@inria.fr}
\address{Equipe PARADYSE, Bureau B211
Centre INRIA Lille Nord-Europe
Park Plaza, Parc scientifique de la Haute-Borne, 40 Avenue Halley B\^atiment B, 59650 Villeneuve-d'Ascq
France}
\date{\today}
\thanks{\textbf{Acknowledgments~:} I would like to thank C. Landim for his help drafting this article, and for our numerous fruitful discussions.}

\begin{document}
\begin{abstract}
Using duality techniques, we derive the hydrodynamic limit for one-dimensional, 
boundary-driven, symmetric exclusion processes with
different types of non-reversible dynamics at the boundary, for which the classical entropy method fails.
\end{abstract}
\pagenumbering{Alph}
\begin{titlepage}
\maketitle
\thispagestyle{empty}
\end{titlepage}
\pagenumbering{arabic}

%\tableofcontents
\section{Introduction}
Boundary-driven exclusion models can provide good examples of simple, solvable non-equilibrium models (see \cite{BD1} and references therein). 
Such processes can exhibit rich behavior, depending on the nature on the boundary dynamics selected. 
One key goal in the study of these models is the derivation for nonequilibrium models of the large deviation functionals which  
plays the role the entropy does for equilibrium models (cf. \cite{Onsager1}, \cite{Onsager2}). Although progress has been made for specific models \cite{BLS2002}, as a first step to achieve this program for general classes of models, it is necessary to study both the stationary 
state and the hydrodynamic behavior of such models.

In \cite{ELX2017}, we investigated the stationary state for three classes of one-dimensional dynamics, 
whose bulk dynamics is symmetric simple exclusion (SSEP), and driven out of equilibrium by non-reversible, 
non-conservative dynamics at the boundaries. 
Classical tools can in some cases be adapted to derive the hydrodynamic limit for boundary-driven one-dimensional models
for which the boundary dynamics is either reversible w.r.t. a product measure, or sped up or slowed down w.r.t. the bulk dynamics (cf. \cite{SlowBond}). 
However, to the best of our knowledge, no hydrodynamic limit has been derived for models whose boundary evolves 
on the same time scale as the bulk and whose dynamics is not reversible with respect to a product measure.

In this article, we expand on the results obtained in \cite{ELX2017}, and use duality techniques to derive the hydrodynamic limit for 
two of the three classes of models investigated in \cite{ELX2017}. Duality properties have been extensively used to derive hydrodynamic behaviors, 
in particular for SSEP dynamics (see \cite{DMP} and references therein). One central challenge in using duality to derive hydrodynamic limits lies in 
closing the discrete difference equations satisfied by the $n$-point correlation functions (e.g. \cite{DMFL} for the Glauber+Kawasaki dynamics). 
In the first class of dynamics studied in this article, the boundary Markov generator 
preserves polynomials of degree one and two in the configuration, and therefore the equations for the density and correlation fields are naturally closed. 
In the second class, particles are created and annihilated  at the left boundary 
at a rate which depends in a weak way on the local configuration at the left boundary (cf. \eqref{eq:Asslambda}). 
This ensures that the dual branching process ultimately dies, and that we are therefore indeed able to close the equations (cf. \cite{ELX2017}). In both cases, 
the bulk dynamics is symmetric simple exclusion, and the right boundary is in contact with a reservoir at density $\beta\in (0,1)$.

Note that, although the method used in this article applies to the third class of models investigated in \cite{ELX2017}, 
in which the left boundary dynamics is sped up by an extra factor $\ell_N\to\infty$, to derive the hydrodynamic limit
they also require $\ell_N$ to be at least of order $N$. 
In this case, however, an adaptation of the more classical \emph{entropy method} \cite{KL} can also be used, 
so that we do not consider this third class here. Because this article is based on duality argument, one main drawback of 
our method is that it is really mainly adapted to models with stirring dynamics in the bulk. Furthermore, the question of what happens when 
the creation/annihilation rate at the left boundary strongly depends on the configuration (i.e. when Assumption \eqref{eq:Asslambda} fails) remains open.
It is to be noted that more recently, a general theory has been developed for the use of duality \cite{GKRV} in the context of interacting particle systems, 
applied in particular to so called inclusion dynamics \cite{OR} and asymmetric exclusion models to derive hydrodynamic limits. 
In this article, however, we do not need such elaborate tools, mainly because our bulk dynamics, namely simple exclusion, 
yields fairly simple equations for the density and correlation fields.

This article is organized as follows; in Section \ref{sec:firstmodel}, we introduce the model, as well as the class of left boundary 
conditions to which our result applies. 
In Section \ref{sec:rho}, we carefully estimate the density at the left boundary, as well as the evolution of the density's gradient at each boundary.
In Section \ref{sec:phi}, we estimate the correlation function of the dynamics. We conclude in Section \ref{sec:Hydro} the proof of our main result using the estimates obtained in
Sections \ref{sec:rho} and \ref{sec:phi}.

\section{Notations and main results}
\label{sec:firstmodel}
\subsection{General notations}
Consider $\Lambda_N=\{1,\dots ,N-1\}$, and let $
\Omega_N=\{0,1\}^{\Lambda_N}.$ our the set of configurations on $\Lambda_N$. Elements of $\Omega_N$ will be denoted $\eta$, and for $j\in \Lambda_N$, $\eta_j=1$ (resp. $0$) is to be understood as site $j$ being occupied (resp. empty) in $\eta$. We study in this article a Markov chain on $\Omega_N$ whose generator can be written
\begin{equation}
\label{eq:DefLN}
L_N=N^2(L_{r,N}+L_{b,N}+L_{l,N}).
\end{equation}
The generator $L_{b,N}$ encompasses the bulk dynamics, symmetric simple exclusion, on $\Lambda_N$. More
precisely, for any function $f:\Omega_N\to\R$, and any configuration $\eta\in \Omega_N$,
\begin{equation*}
\label{eq:DefLbulk}
(L_{b,N}f)(\eta)=\sum_{k=1}^{N-2}\{f(\sigma^{k,k+1}\eta)-f(\eta)\},
\end{equation*}
where $\sigma^{k,l}\eta$ is the configuration obtained from $\eta$ by
swapping the occupation variables $\eta_k$, $\eta_l$, 
\[(\sigma^{k,l}\eta)_j=\begin{cases}
\eta_{l} & \mbox{ if } j=k\\
\eta_k & \mbox{ if } j=l\\
\eta_j & \mbox{ if } j\in \Lambda_N\setminus \{k,l\}
\end{cases}.\]

At both boundaries, the dynamics is put in contact with non-conservative dynamics. On the right, the dynamics is coupled to a reservoir at density $\beta\in(0,1)$
\begin{equation*}
\label{eq:DefLrb}
(L_{r,N}f)(\eta)=[\beta(1-\eta_{N-1})+(1-\beta)\eta_{N-1}]
\br{f\pa{\sigma^{N-1}\eta}-f(\eta)},
\end{equation*}
where for $k\in \Lambda_N$, $\sigma^k\eta$ is the configuration where
the state of site $k$ has been flipped,
\[(\sigma^k\eta)_j=\begin{cases}
1-\eta_k & \mbox{ if } j=k\\
\eta_j & \mbox{ if } j\in \Lambda_N\setminus \{k\}
\end{cases}.\]
Note that we choose at the right boundary a very simple dynamics (coupling with 
a large reservoir at equilibrium). However, our method still applies if the right boundary generator is chosen according to either of the two classes of dynamics introduced below.

Fix $p\in \N$, we denote $\lps=\{1,\dots,p\}$ the microscopic set that plays the role of left boundary for $\Lambda_N$. The left boundary generator is written $L_{l,N}=L_R+L_C+L_A$, where
\begin{equation}
\label{eq:DefLR}
(L_{R}f)(\eta)=\sum_{j\in \lps}r_j\cro{\alpha_j(1-\eta_j)+
\eta_j(1-\alpha_j)}\br{f\pa{\sigma^{j}\eta}-f(\eta)},
\end{equation}
\begin{equation}
\label{eq:DefLC}
(L_{C}f)(\eta)=\sum_{j\neq k\in \lps }c_{j,k}\cro{\eta_j(1-\eta_k)+\eta_k(1-\eta_j)}
\br{f\pa{\sigma^j\eta}-f(\eta)},
\end{equation}
\begin{equation}
\label{eq:DefLA}
(L_{A}f)(\eta)=\sum_{j\neq k\in \lps }a_{j,k}\cro{\eta_j\eta_k+(1-\eta_j)(1-\eta_k)}
\br{f\pa{\sigma^j\eta}-f(\eta)}
,\end{equation}
and $(r_j)_{j\in \lps}$, $(c_{j,k})_{j\neq k\in \lps}$, $(a_{j,k})_{j\neq k\in \lps}$  are non-negative constants. 
The $(\alpha_j)_{j\in\lps}$ are in $[0,1]$, and are the respective densities of each of the reservoirs linked to sites $1\leq j\leq p$. 
The $c_{j,k}$'s (resp. $a_{j,k}$'s) are to be understood as copy (resp. anticopy) rates, at which site $j$ takes the value 
(resp. the inverse of the value) of site $k$. The $r_j$'s are reservoir rates, at which site $j$ is updated according to a reservoir at density $\alpha_j$. 
Note that the stirring generator \eqref{eq:DefLbulk} occurs in $\Lambda_N$, therefore it also affects the left boundary $\lps$, as well as links $\lps$ with $\Lambda_N\setminus \lps=\{p+1, \dots, N-1\}$.
We prove in Lemma 3.3 of \cite{ELX2017} that, assuming 
\begin{equation}
\label{eq:assrates}\tag{A1}
\sum_{j\neq k\in \lps }a_{j,k}+\sum_{j\in \lps }r_j>0, 
\end{equation}
the generator $L_{l,N}+L_{b,p+1}$ (where $(L_{b,p+1}f)(\eta)=\sum_{k=1}^{p-1}\{f(\sigma^{k,k+1}\eta)-f(\eta)\}$ is the stirring generator limited to jumps in $\lps$) 
admits a unique invariant measure $\mu$ (which does not depend on $N$). We denote 
\begin{equation}
\label{eq:Defalpha}
\alpha:=\E_{\mu}(\eta_p).
\end{equation}
As investigated in \cite{ELX2017}, the non-conservative dynamics encoded in $L_{l,N}$ macroscopically behaves as a reservoir at density $\alpha$. 
We will not consider the case $\sum_{j\neq k\in \lps }a_{j,k}+\sum_{j\in \lps }r_j=0$, $\sum_{j\neq k\in \lps }c_{j,k}>0$, in which $L_{l,N}+L_{b,p}$ admits 
two degenerate stationary states respectively concentrated on the full and empty configurations. 

In this article, we will focus on the case 
\[\sum_{j\neq k\in \lps }a_{j,k}=0, \quad \mbox{ and } \quad   \sum_{j\in \lps }r_{j}>0.\] 
This is purely for convenience: the case $\sum_{j\neq k\in \lps }a_{j,k}>0$ offers no further difficulty w.r.t. the hydrostatic limit, so that the hydrodynamic 
limit in this case can be quite easily recovered from the present article and the tools introduced in \cite{ELX2017}. 

\medskip

Fix a smooth initial density profile
$\rho_0\in C^2([0,1])$, and denote by $\nu_N$ the
product measure on $\Omega_N $ close to the profile $\rho_0$
\begin{equation*}
\nu_N(\eta)=\prod_{k\in\Lambda_N}
\cro{\eta_k\rho_0(k/N)+(1-\eta_k)(1-\rho_0(k/N))}.
\end{equation*}

Let $D(\bb R_+,\Omega_N)$ the space of right-continuous functions
$\eta: \bb R_+ \to \Omega_N$ with left limits.  Denote by
$\Prob_{\nu_N}$ the distribution on $D(\bb R_+,\Omega_N)$ induced by the
process $\eta(t)$ started from $\nu_N$, and driven by the generator $L_N$. Expectation with respect to
$\Prob_{\nu_N}$ is denoted $\E_{\mu^N}$. We are now ready to state our main result.

\begin{theorem}
\label{thm:HydroGen}
Fix $T>0$, and assume that \eqref{eq:assrates} holds. For any continuous function $G:[0,1] \to \bb R$, and any $t\in [0,T]$
\begin{equation*}
\lim_{N\to \infty} \E_{\nu_N}\pa{\abs{\frac 1N \sum_{k\in \Lambda_N} G(k/N) \, \eta_k(t) -
		\int_{[0,1]} G(u) \bar \rho (t,u)\, du}}=0\;,
\end{equation*}
where $\bar \rho$ is the unique solution of the linear elliptic equation
\begin{equation}
\label{eq:Eqrhobar}
\begin{cases}
\partial_t\rho(t,u)=\Delta \rho(t,u)  & \mbox{ for any } (t,u)\in]0,T]\times]0,1[  \\
\rho(0,\cdot)=\rho_0(\cdot)&\\
\rho(t,0) \,=\, \alpha, \quad \rho(t,1) \,=\, \beta&  \mbox{ for any }  t\in]0,T] 
\end{cases}\end{equation}
where $\alpha$ was defined in \eqref{eq:Defalpha}.
\end{theorem}

\begin{remark}
As shown in \cite{ELX2017}, this choice for the left boundary generator $L_{l,N}$ is the most general for which we can write 
\[L_{l,N}\eta_j=q^{1,j}+\sum_{k\in \lps}q^{1,j}_k\eta_k\]
and 
\[L_{l,N}\eta_j\eta_k=q^{2,j,k}+\sum_{l\in \lps}q^{2,j,k}_l\eta_l+\sum_{l\neq m\in \lps}q^{2,j,k}_{l,m}\eta_l\eta_m,\]
for some constants $q^{1,j}$, $q^{1,j}_k$, $q^{1,j}$, $q^{2,j,k}$ $q^{2,j,k}_{l}$ and $q^{2,j,k}_{l,m}$. In other words, this model is the most general for which $L_{l,N}$ preserves polynomials of degree $\nu$ in $\eta$.

When this condition is not respected, one can still derive a hydrodynamic limit, if at the left boundary, particles are created and removed at a rate which depends in a small measure on the configuration at the boundary. This is the content of the next section.\end{remark}

\subsection{Creation/annihilation rate depending on the local boundary configuration} 
\label{sec:Model2}
In order to present as general a result as possible, we now change the left boundary generator $L_{l,N}$, to one where particles are created and annihilated at the first site depending on the state of the boundary. We therefore let
\begin{equation}\label{eq:genl2}
(\widetilde{L}_{l,N} f)(\eta) =c(\eta_1,\dots,\eta_p) \br{f(\sigma^{1}\eta)-f(\eta)}\;.
\end{equation}
where $c$ is a function $c:\{0,1\}^p\to \R_+$. 
 
 \bigskip
 
The dynamics for this model is more general, however in order to derive the hydrodynamic limit, we need to assume that the creation and annihilation rate $c$ do not depend too much on the boundary configuration. Let us denote by $\xi$ the elements of $\{0,1\}^{p-1}$, we let 
\[A=\inf_{\xi\in\{0,1\}^{p-1}}c(0, \xi)\]
\[B=\inf_{\xi\in\{0,1\}^{p-1}}c(1,\xi)\]
the minimal creation and annihilation rates. Denote $\lambda(0,\xi)=c(0,\xi)-A$ and $\lambda(1,\xi)=c(1,\xi)-B$, we assume that 
\begin{equation}\label{eq:Asslambda}\tag{A2}(p-1)\sum_{\xi\in\{0,1\}^{p-1}}\br{\lambda(0,\xi)+\lambda(1,\xi)}\leq A+B .\end{equation}
We now state the hydrodynamic limit for this second model.
We use analogous notations as for Theorem \ref{thm:HydroGen}, and denote
$\widetilde{\Prob}_{\nu_N}$ the distribution on $D(\bb R_+,\Omega_N)$ induced by the
process $\eta(t)$ started from $\nu_N$, and driven by the generator $\widetilde{L}_N:=\widetilde{L}_{l,N}+L_{b,N}+L_{r,N}$. Expectation with respect to
$\widetilde{\Prob}_{\nu_N}$ is denoted $\widetilde{\E}_{\mu^N}$.
\begin{theorem}
	\label{thm2}
	Assume \eqref{eq:Asslambda}, there exists $\widetilde{\alpha}\in [0,1]$ such that for any $T>0$, any continuous function $G:[0,1] \to \bb R$, and any $t\in [0,T]$
\begin{equation*}
\lim_{N\to \infty} \widetilde{\E}_{\nu_N}\pa{\abs{\frac 1N \sum_{k\in \Lambda_N} G(k/N) \, \eta_k(t) -
		\int_{[0,1]} G(u) \bar \rho (t,u)\, du}}=0\;,
\end{equation*}
where $\bar \rho$ is the unique solution of \eqref{eq:Eqrhobar}, except with $\widetilde{\alpha}$ replacing $\alpha$. 
\end{theorem}
Note that the left density $\widetilde{\alpha}$ is the limit $\alpha$ introduced in Theorem 2.4 of \cite{ELX2017}. We will only write the proof of Theorem \ref{thm:HydroGen} and assume that $L_A=0$. With the tools developed for the hydrostatic limit in \cite{ELX2017}, the proof of Theorem  \ref{thm:HydroGen} extends straightforwardly to both the case $\sum_{j\neq k\in \lps}a_{j,k}=0$ and Theorem \eqref{thm2}.

\subsection{Duality and scheme of the proof}
Denote for any $(t,k,l)\in [0,T]\times \Lambda_N^2$
\begin{equation}
\label{eq:Defrho}
 \rho_N(t,k)=\E_{\nu_N}(\eta_k(t))
\end{equation}
the density at site $k\in \lambda_N$, and adopt a similar notation for the two-points correlation function
\begin{equation}
\label{eq:DefCorHydro}
\varphi_N(t,k,l)=\E_{\nu_N}\Big(\{\eta_k(t)-\rho_N(t,k)\}\{\eta_l(t)-\rho_N(t,l)\}\Big).
\end{equation}
To prove Theorem \ref{thm:HydroGen}, we will use duality between $\rho_N$, (resp. $\varphi_N$), and random walks on $\Lambda_N$ (resp. $\Lambda_N^2$).

\medskip

We start by introducing a set of cemetery states 
\[\partial\Lambda_N=\{\dd_1,\dots,\dd_p\}\cup\{N\},\]
each representing one of the reservoirs, and let $\bar\Lambda_N=\Lambda_N\cup\partial\Lambda_N$. Further define the function ${\rho_\dd}$ on $\partial \Lambda_N$ given by 
\[\rho_\dd(\dd_j)=\alpha_j, \; \forall j\in \lps \quad \mbox{ and } \quad   \rho_\dd(N)=\beta.\]
We extend the function $\rho_N $ defined in \eqref{eq:Defrho} to $[0,T]\times \bar\Lambda_N$ by letting for any $ t\geq 0$
\begin{equation}
 \label{eq:bcrho}
 \rho_N(t,\cdot)=\rho_\dd(\cdot) \quad \mbox{ on }\;\partial\Lambda_N.
\end{equation}

\medskip

We now introduce dual generators, acting on functions on $\bar\Lambda_N$, defined by 
\begin{equation}
\label{eq:DefLbulkd}
(\gend_{b,N}f)(j)=\begin{cases}
                  (\Delta_Nf)(j):= f(j+1)+f(j-1)-2f(j) & \mbox{ for }1<j<N,\\
		   ( \nabla^+_N f) (j):=f(j+1)-f(j)& \mbox{ for }j=1,\\
		    (\nabla^-_N f) (j):=f(j-1)-f(j)& \mbox{ for }j=N-1,
		    \end{cases},
\end{equation}
\begin{equation*}
\label{eq:DefLrbd}
(\gend_{r,N}f)(j)=\ind{j=N-1}\br{f(N)-f(j)}
\end{equation*}
\begin{equation}
\label{eq:DefLRd}
(\gend_{R}f)(j)=\ind{j\in \lps}r_j\br{f(\dd_j)-f(j)}
\end{equation}
and finally
\begin{equation}
\label{eq:DefLCd}
(\gend_{C}f)(j)=\ind{j\in \lps}\sum_{k\neq j\in \lps }c_{j,k}\br{f(k)-f(j)}.
\end{equation}
Note in particular that any of the cemetery states in $\partial\Lambda_N$ is an absorbing state for each of these dual generators.
Then, letting 
\begin{equation}
\label{eq:Defgend}
\gend_N=\gend_{R}+\gend_{C}+\gend_{b,N}+\gend_{r,N}, 
\end{equation}
using the fact that $\partial_t\rho_N(t,k)=N^2E_{\nu_N}(L_N\eta_k(t))$ and notation \eqref{eq:bcrho}, one obtains after elementary computations that the function $\rho_N$ defined in \eqref{eq:Defrho} is a solution of the system 
\begin{equation}
\label{eq:systrho}
\begin{cases}
   \partial_tf=N^2 \gend_N f&\\
   f(0,\cdot)=\rho_0(\cdot/N)& \mbox{ on } \Lambda_N\\
   f(t,\cdot)=\rho_\dd(\cdot),& \mbox{ on }\partial \Lambda_N ,\; \forall t\in[0,T]
  \end{cases}. 
\end{equation}

The first ingredient to prove Theorem \ref{thm:HydroGen} is showing that for any $t=t(N)$ large enough,
\begin{equation}
\label{eq:rhop1}
\rho_N(t,p+1)=\alpha+o_N(1),
\end{equation}
where $\alpha$ is given by \eqref{eq:Defalpha}. Since $\rho_N$ is solution of \eqref{eq:systrho}, and since on $\{p+1,\dots,N-1\}$ $\gend_N$ acts as the discrete Laplacian $\Delta_N$ (with our notation for site $N$, $(\{\gend_{r,N}+\gend_{b,N}\}f)(N-1)=(\Delta_N f)(N-1)$), this yields
\begin{equation} 
\label{eq:eqrhoapprox}
\begin{cases}
\partial_t \rho_N(t,k)=N^2(\Delta_N \rho_N)(t,k) &\forall (t,k)\in [0,T]\times\{p+2,N-1\}\\
\rho_N(0,\cdot)=\rho_0(\cdot/N)&\mbox{ on }\{p+2,\dots,N-1\}\\
\rho_N(t,p+1)=\alpha +o_N(1)&  \forall t\in]0,T]\\
\rho_N(t,N)=\beta&  \forall t\in]0,T],
\end{cases}
\end{equation}
whose solution converges weakly as $N\to\infty$ towards the solution of \eqref{eq:Eqrhobar}. Proving \eqref{eq:rhop1} is the purpose of Section \ref{sec:rho}.

\medskip

The second ingredient is a control of the two-points correlation function $\varphi_N$ defined in \eqref{eq:DefCorHydro}: for some large set $S_{N,\delta} \subset \{(k,l), \quad p+1\leq k<l\leq N-1\}$
\[ \limsup_{N\to\infty}\sup_{\substack{(k,l)\in S_{N,\delta}\\
t\in [0,T]}}\abs{\varphi_N(t,k,l)}= 0,\]
which allows, in Theorem \ref{thm:HydroGen}, to replace $\eta_k(t)$ by its expectation $\rho_N(t,k)$. This estimate is obtained in Section \ref{sec:phi}. With these two elements, a few technical difficulties remain to prove Theorem \ref{thm:HydroGen}, which is done in Section  \ref{sec:Hydro}.

\section{Estimation of the left density and technical lemmas}
\label{sec:rho}
\subsection{Estimation of the density at the boundaries}
Define a continuous time random walk $\widetilde{X}$ on $\bar\Lambda_N$ driven by the sped-up dual generator $N^2\gend_N$ defined in \eqref{eq:Defgend}. For any set $B\subset \bar \Lambda_N$, define $\widetilde{H}(B)$ as $\widetilde{X}$'s hitting time  of the set $B$, 
\[\widetilde{H}(B)=\inf\{s\geq 0, \widetilde{X}\in B \},\]
and define $\widetilde{H}_t(B)=t\wedge \widetilde{H}(B)$. For any $ j\in \bar \Lambda_N$, denote $\widetilde{\Prob}_j$ and $\widetilde{\E}_j$ the distribution of $\widetilde{X}$ started from $j$ and its expectation. Since $\rho_N$ is solution of \eqref{eq:systrho}, it is well known that for any $j\in \bar \Lambda_N$ and any $t\geq 0$
\begin{equation*}
\rho_N(t,j)=\widetilde{\E}_{j}\cro{b_N\big(t-\widetilde{H}_t,\widetilde{X}(\widetilde{H}_t)\big)},
\end{equation*}
where we shortened $\widetilde{H}_t:=\widetilde{H}_t(\partial \Lambda_N)$ and where $b_N$ is the function giving the value of $\rho_N$ at the space-time boundary, defined as
\begin{equation}
b_N(t,j)=\rho_0(j/N)\ind{ t=0, j\in \Lambda_N}+\rho_\dd(j)\ind{j\in \partial\Lambda_N}.
\end{equation}

\medskip

To present the proof in as simple a setting as possible, however, it is not convenient that $\gend_N$ has absorbing states. We therefore define $ \gendb_N=\gend_N+\gend_{\dd,N}$, with
\begin{equation*}
(\gend_{\dd,N}f)(j)=\ind{j\in \partial\Lambda_N}\br{f(p+1)-f(j)}
,\end{equation*}
which allows jumps at rate $1$ from any of the cemetery states to  site $p+1$. We denote $X$ a random walk driven by $N^2\gendb_N$, in particular, assuming that both random walk start from the same point in $\Lambda_N$, $X$ coincides with $\widetilde{X}$ at least up until time
\[H(\partial\Lambda_N)=\inf\{s\geq 0, \; X\in \partial\Lambda_N \}=\widetilde{H}(\partial\Lambda_N).\]
We denote without "$\;\widetilde{\;\;}\; $" all quantities and items relative to $X$, and thanks to the last observation, we are still able to write
\begin{equation}
\label{eq:dualrho} 
\rho_N(t,j)=\E_{j}\cro{b_N\big(t-H_t,X(H_t)\big)},
\end{equation}
where once again $H_t:=H_t(\partial \Lambda_N)$

We are now ready to state the main result of this section. We start by proving equation \eqref{eq:rhop1}, in order to obtain the differential system \eqref{eq:eqrhoapprox}. Since we will also need to control the density gradient at the boundaries $p+1$ and $N-1$ to estimate the correlations, we prove a more general result than \eqref{eq:rhop1}, and estimate carefully the left and right densities.
\begin{proposition}
\label{lem:convrhop}
For any $\varepsilon>0$, there exists a constant $K=K(\varepsilon)$ such that for any time $s\in ]0,T]$ , 
\begin{equation}
\label{eq:convrhop1}
\sup_{t\in [s, T]}\abs{\rho_N(t,p+1) -\alpha}\leq K \frac{N^{\varepsilon-1}}{s},
\end{equation}

	\begin{equation}
	\label{eq:convrhop2}
\sup_{t\in [s,T]}\abs{\rho_N(t,p+1)-\rho_N(t,p+2)} \leq K \frac{N^{\varepsilon-1}}{s},
	\end{equation}
and
\begin{equation}
\label{eq:convrhop3}
\sup_{t\in [s,T]}\abs{\beta-\rho_N(t,N-1)}\leq  K \frac{N^{\varepsilon-1}}{s}.
\end{equation}
\end{proposition}
\begin{proof}[Proof of Proposition \ref{lem:convrhop}]We will only detail the proof for the first identity, since the second is  and third are proved in the same way. 
To estimate $\rho_N(t,p+1)$, we use \eqref{eq:dualrho}, and start the random walk $X$ at site $p+1$. Then, it performs excursions away from $p+1$, either in the bulk, in which case the excursion lasts a macroscopic time of order $1/N$ (recall that the whole random walk is accelerated by $N^2$), or in the left boundary, in which case it has a positive probability $\pi$ to reach one of the cemetery states $\dd_j$'s. In a time $s$, $X$ will perform a number of excursions at least of order $sN$, each yielding a chance of ending in one of the cemetery states.
	
	\medskip
	
	We now make this argument rigorous. Since similar proofs will be used repeatedly, we detail the proof of this Proposition, and will be more concise later on. Recall that $X$ can jump from $\partial\Lambda_N$ to $p+1$ at rate $1$, let $t_0=0<t_1<t_2\dots$ denote the successive arrival times of $X$ at the site $p+1$ :
\[t_0=0,\quad \mbox{ and } \quad t_{n+1}=\inf\{ t>t_n,\quad  X(t)=p+1\quad \mbox{and}\quad X(t^-)\neq p+1\}.\]
The random walk $X$ being a Markov process, under $\Prob_{p+1}$, the successive excursions  $(X(t))_{t_n\leq t<t_{n+1} }$ are i.i.d. in $n\geq 0$. To distinguish the two types of excursions away from $p+1$, denote 
\[E_n=\{X(t_{n+1}^-)=p+2\}\]
which indicates the $n$-th excursion was performed in the bulk rather than in the boundary. Let us denote 
\[F_n=\{\mbox{ There exists }(t,j)\in [t_n,t_{n+1}[\times \lps, \mbox{ such that }X(t)=\dd_j\}.\]
\[(\mbox{resp. }\; G_n=\{\mbox{ There exists }t\in [t_n,t_{n+1}[, \mbox{ such that }X(t)=N\}\;),\]
which indicates the $n$-th excursion reached one of the cemetery states $\dd_n$ (resp. $N$).
Finally, we denote by $d_n=t_{n+1}-t_n$ the duration of the $n$-th excursion. Since the excursions away from $p+1$ are i.i.d. under $\Prob_{p+1}$, so are the $(E_n)_{n\in \N}$, $(F_n)_{n\in \N}$,  $(G_n)_{n\in \N}$, and $(d_n)_{n\in \N}$.

\bigskip
Define 
\begin{equation}
\label{eq:Defpialpha1}
\pi=\Prob_p\Big[H(\partial\Lambda_N)<H(\{p+1\})\Big],
\end{equation} 
which is the probability that an excursion in the left boundary reaches one of the cemetery states before coming back to site $p+1$.
One easily obtains, for any $n\geq 0$, the following properties :
 \begin{equation}
 \label{eq:P1}
 \quad \Prob_{p+1}(E_n)=\Prob_{p+1}(E_n^c)=1/2,
 \end{equation}
\begin{equation}
\label{eq:P2}
\Prob_{p+1}(F_n\mid E_n)=0\quad \mbox{ and } \quad \Prob_{p+1}(F_n\mid E_n^c)=\pi,
\end{equation}
\begin{equation}
\label{eq:P3}
\Prob_{p+1}(G_n\mid E_n)=\frac{1}{N-1-p}\quad \mbox{ and } \quad \Prob_{p+1}(G_n\mid E_n^c)=0.
\end{equation}
Furthermore, there exists a constant $C$ such that  
\begin{equation}
\label{eq:P4}
\E_{p+1}(d_n\mid E_n)\leq C/N \quad \mbox{ and } \quad \E_{p+1}(d_n\mid E_n^c)\leq C/N^2.
\end{equation}
Because of the second part of this equation, the constant $C$ depends a priori on the rates of the left boundary dual generator $\gend_{l,N}$.
These identities are elementary, we do not detail their proof. In particular, \eqref{eq:P4} uses the fact that the generator $\gend_N$ was accelerated by $N^2$, therefore a excursion in the bulk as a duration of order $1/N$, whereas in the boundary, the random walk will perform a finite number of steps before heading back to site $p+1$, so that the time length of a typical excursion is of order $N^{-2}$.

\medskip

We now prove \eqref{eq:convrhop1}. Fix $t\in [0,T]$, \eqref{eq:dualrho} yields 
\begin{multline*}
\rho_N(t,p+1)=\E_{p+1}\cro{b_N\big(t-H_t(\partial\Lambda_N),X(H_t(\partial\Lambda_N))\big)}\\
=\E_{p+1}\Big[\rho_0(X(t/n))\ind{H_t(\partial\Lambda_N)=t}\Big]+\beta\Prob_{p+1}\Big[H_t(\partial\Lambda_N)=H(\{N\})\Big]\\
+\sum_{j\in \lps}\alpha_j\Prob_{p+1}\Big[H_t(\partial\Lambda_N)=H(\{\dd_j\})\Big],
\end{multline*}
so that, since $\rho_0$, $\beta$ and the $\alpha_j$'s are less than $1$, 
\begin{multline}
\label{eq:boundrhop}
\abs{\rho_N(t,p+1)-\sum_{j\in \lps}\alpha_j\Prob_{p+1}\Big[H(\partial\Lambda_N)=H(\{\dd_j\})\Big |H(\partial\Lambda_N\setminus\{N\})\leq H(\{N\})\Big]}\\
\leq C\Prob_{p+1}\Big[H_t(\partial\Lambda_N)=H_t(\{N\})\Big]
\end{multline}
for some constant $C:=C(p)$.
\medskip

Let us now estimate the right-hand side above. From equations \eqref{eq:P1}, \eqref{eq:P2} and \eqref{eq:P3}, we obtain immediately, since $F_n$ and $G_n$ are disjoint events, that for any $n\in \N$
\begin{equation}
\label{eq:PFnGn}
\Prob_{p+1}(F_n\cup G_n)=\frac{1}{2}\pa{\pi+\frac{1}{N-1-p}}:=\delta>0.
\end{equation}
Fix $0<t<T$, and denote by $M=M(t)$ the number of complete excursions occurring before $t$,
\[M=\max\{n\in \N ,\; t_n<t \mbox{ and } t_{n+1}\geq t \}.\]
By definition of $F_n$ and $G_n$, we have 
\[\{H_t(\partial\Lambda_N)=t\}\subset\cap_{n=0}^M (F_n\cup G_n)^c,\] 
so that for any $m\in \N$
\begin{align*}
\Prob_{p+1}\Big[H_t(\partial\Lambda_N)=t\Big]&\leq \Prob_{p+1}\pa{\bigcap_{n=0}^m (F_n\cup G_n)^c}\Prob_{p+1}(M\geq m)+\Prob_{p+1}(M\leq m)\\
&\leq \Prob_{p+1}\pa{\bigcap_{n=0}^m (F_n\cup G_n)^c}+\Prob_{p+1}\pa{\max_{0\leq k\leq m} d_{k}\geq \frac{t}{m+1}}\\
&\leq (1-\delta)^m+(m+1)\Prob_{p+1}\pa{d_0\geq \frac t {m+1}}\\
\end{align*}
thanks to equation \eqref{eq:PFnGn}.
According to \eqref{eq:P4}, $\E_{p+1}(d_
0)\leq C/N$. By Markov inequality, the second term is thus less than $C(m+1)^2/tN$. Since $\delta>\pi/2$, we then let $m=-\log N/\log(1-\pi/2)$ to obtain that that for some constant $K_1$ depending on $T$, $C$ and $ \pi$,
\begin{equation}
\label{eq:boundt}
\Prob_{p+1}\Big[H_t(\partial\Lambda_N)=t\Big]\leq  \frac{K_1(\log N)^2}{tN}.
\end{equation}
Furthermore, using equation \eqref{eq:P2} and \eqref{eq:P3},
\begin{align}
\label{eq:boundN}
\Prob_{p+1}\big[H_t(\partial\Lambda_N)=H(\{N\})\big]&\leq \Prob\big[H(\{N\})<H(\partial\Lambda_N\setminus\{N\}) \big]\nonumber\\
&=\frac{1}{2\delta(N-1-p)}\leq\frac{1}{\pi(N-1-p)}.
\end{align}

\medskip

Using  \eqref{eq:boundt}, and the bound above, we thus obtain that for any $\varepsilon>0$, there exists a constant $K_2$ depending on $T$, $C$, $\pi$ and $p$, such that 
\begin{equation*}
\label{eq:probmort}
\Prob_{p+1}(H_t(\partial\Lambda_N)=H_t(\{N\}))\leq K_2\frac{N^{\varepsilon-1}}{s} ,\end{equation*}
for any $t\in[s,T]$. Letting $K_3=C(p)K_2$, for any $\varepsilon>0$, we obtain from equation \eqref{eq:boundrhop} that for any $ t\in[s ,T]$
\begin{multline}
\label{eq:sumprobalpha}
\abs{\rho_N(t,p+1)-\sum_{j\in \lps}\alpha_j\Prob_{p+1}\Big[H(\partial\Lambda_N)=H(\{\dd_j\})\Big |H(\partial\Lambda_N\setminus\{N\})\leq H(\{N\})\Big]}\\
\leq K_3\frac{N^{\varepsilon-1}}{s}. 
\end{multline}

Let us denote $Y$ a random walk started from $\lps$, and driven by the generator $\gend_{l,N}+\gend_{b,p+1}$, where $(\gend_{b,p+1}f)(j)$ is defined in \eqref{eq:DefLbulkd} as the generator of a symmetric random walk on $\lps$ with reflection boundary conditions. Denote ${\mathbb Q}_{j}$ the distribution of $Y$ started from $j$, and $H^Y(B)$ the hitting time of the set $B$ by $Y$. Since the only cemetery states that can be reached by $Y$ are the $\dd_j$'s, the Markov property yields 
\begin{multline}
\label{eq:idQP}
\Prob_{p+1}\Big[H(\partial\Lambda_N)=H(\{\dd_j\})\Big |H(\partial\Lambda_N\setminus\{N\})\leq H(\{N\})\Big]\\
={\mathbb Q}_{p}\big[H^Y(\partial\Lambda_N\setminus \{N\})=H^Y(\{\dd_j\})\Big]. 
\end{multline}
Recall that we denoted $\mu$ the unique invariant measure of the generator $L_{l,N}+L_{b,p+1}$, letting 
\[\rho^*(j)=\E_\mu(\eta_j)\ind{j\in \Lambda_N}+\alpha_j\ind{j\in \partial\Lambda_N\setminus\{N\}},\] 
elementary computations similar to those performed for $\rho_N$ yield that the function $\rho^*$ is solution on $\lps\cup\partial \Lambda_N\setminus\{N\}$ of $(\gend_{l,N}+\gend_{b,p+1})\rho^*=0$ with boundary condition $\rho^*(\dd_j)=\alpha_j$.
By duality, we can therefore write  
\[\sum_{j\in \lps}\alpha_j\mathbb Q_{p}\Big[H^Y(\partial\Lambda_N\setminus \{N\})=H^Y(\{\dd_j\})\Big]=\E_\mu(\eta_p)=\alpha,\] 
where  $\alpha$ was introduced in \eqref{eq:Defalpha}. Recalling \eqref{eq:idQP} then allows \eqref{eq:sumprobalpha} to be rewritten as wanted
\[\sup_{t\in[s ,T] }\abs{\rho_N(t,p+1)-\alpha}\leq K_3\frac{N^{\varepsilon-1}}{s}.\]

\bigskip

We now turn to equation \eqref{eq:convrhop2}. With only minimal adaptation of the proof above, we can write for a larger constant $K_4$
\[\sup_{t\in[s ,T] }\abs{\rho_N(t,p+2)-\alpha}\leq K_4\frac{N^{\varepsilon-1}}{s}.\]
so that equation \eqref{eq:convrhop2} follows immediately from the triangular inequality.

\bigskip

Finally, the third identity \eqref{eq:convrhop3} is proved in the same way as well : we split the random walk started from $N-1$ into excursions away from $N-1$. Each excursion has a probability $1/2$ of ending at site $N$, where the density is $\beta$, and has a probability $1/2(N-1-p)$ of reaching the other boundary. Since the proof is an easier version of that of equation \eqref{eq:convrhop1}, we do not detail it here. 
	\end{proof}
\begin{remark}[Regarding the assumption $\sum_{j\neq k}a_{j,k}=0$]
The step we just performed is the only point in the proof where we used $\sum_{j\neq k}a_{j,k}=0$. If the anticopy generator is added, the dual generators must be defined on the set $\{-1,1\}\times\bar \lambda_N$ instead of $\bar \Lambda_N$, because one must keep track of the number of times the anticopy generator $L_A$ inverted the value of the site occupied by the random walker. This burdens substantially the notations, therefore we refer the interested reader to \cite{ELX2017} for more details on how to overcome this difficulty. 
\end{remark}

	\subsection{Estimation of the gradient at the boundary}
	Now that we have estimated the density at the boundaries, we estimate the gradients at the boundary, which will be needed later on to estimate the correlations. This is done in Lemma \ref{lem:grad} below.
First, we estimate uniformly in its starting point $k\in\{p+1,\dots, N-1\}$, the probability that a random walk on $\bb Z$ reaches either $p+1$ or $N-1$  for the first time in a small time window $[t-s,t]$. Let $Y$ now denote a continuous time random walk in $\Z$ with jump rate $N^2$ to each neighbor, started from $k\in \{p+1,\dots,N-1\}$, and driven by  the generator $N^2\Delta_{\Z}$, where for any function $f:\Z\to\bb R$,
\[(\Delta_{\Z}f)(j)=f(j+1)+f(j-1)-2f(j).\]
As before, let $H^Y(B)$ be the time at which $Y$ reaches the set $B\subset \Z$, 
\[H^Y(B)=\inf\{t\geq 0,\; Y(t)\in B\}.\]
	To keep notations simple, we also denote $\Prob_k$ the distribution of this random walk started at site $k$.
\begin{lemma}
	\label{lem:BoundaryProba1} There exists a constant $C$ such that for any $0<s< t$   
	\[\sup_{k\in \{p+1,\dots,N-1\}}\Prob_k\pa{ H^Y(\{p+1, N-1\})\in [t-s,t]}\leq  C\pa{\frac{s}{t^{3/2}}+\frac{1}{N\sqrt{t-s}}}.\]
	The second term is the error when approximating $Y$ with a Brownian Motion, whereas the first one is the probability above applied to a rescaled Brownian Motion.
\end{lemma}
\begin{figure}
	\input{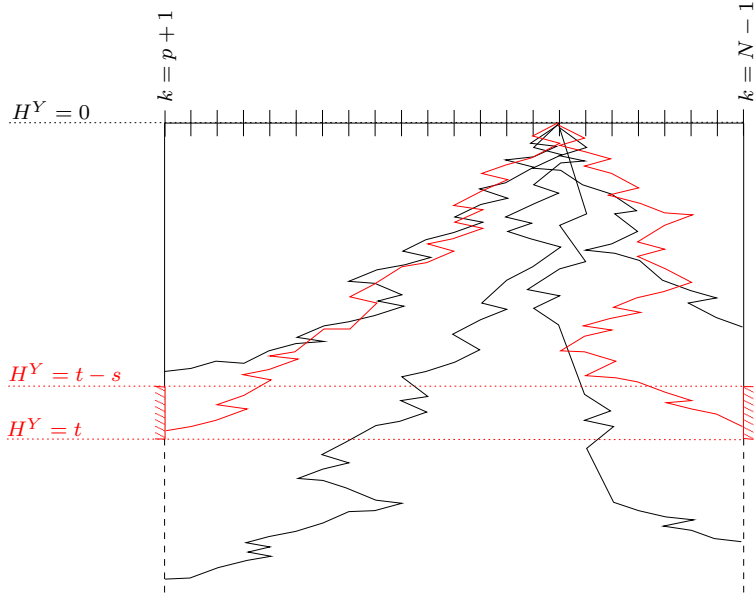}
	\caption{Lemma \ref{lem:BoundaryProba1} estimates the probability of the red trajectories uniformly in the starting point of the random walk.}
	\label{fig:Target}
\end{figure}

\begin{proof}[Proof of Lemma \ref{lem:BoundaryProba1}]
	A visual representation of the Lemma is given in Figure \ref{fig:Target}.
	We want to estimate uniformly in $k\in \{p+1, \dots, N-1\}$ the probability that a random walk started from $k$ hits $p+1$ or $N-1$ between times $t-s$ and $t$. We first write 
	\begin{multline*}\Prob_k\Big[ H^Y(\{p+1, N-1\})\in [t-s,t]\Big]\\
	\leq\Prob_k\Big[  H^Y(\{N-1\})\in [t-s,t]\Big]+\Prob_k\Big[ H^Y(\{p+1\})\in [t-s,t]\Big]\end{multline*}
	The two probabilities on the left hand side are estimated in the same fashion, so that we only estimate the first one. To prove Lemma \ref{lem:BoundaryProba1}, it is therefore sufficient to prove that for  some constant $C$, and any $p+1\leq k\leq N-1$
	\begin{equation}
	\label{eq:EstimatePi}
	\Prob_k\Big[ H^Y(\{N-1\})\in [t-s,t]\Big]\leq  C\pa{\frac{s}{t^{3/2}}+\frac{1}{N\sqrt{t-s}}}
	\end{equation}
	
	Let us denote \[Z(t)=\max_{t\in[0,T]}Y(t),\]
	By reflexion principle
	\begin{multline*}\Prob_k\Big[  H^Y(\{N-1\})\in [t-s,t]\Big]=\Prob_k\Big[Z(t)\geq N-1\Big]-\Prob_k\Big[Z(t-s)\geq N-1\Big]\\
	=2\Prob_0\Big[Y(t)\geq N-1-k\Big]-2\Prob_0\Big[Y(t-s)\geq N-1-k\Big].\end{multline*}
	Since $Y$ is a random walk sped up by $N^2$, the family of increments $(Y((k+1)/N^2)-Y(k/N^2))_{k=0,...,tN^2-1}$ is i.i.d. and each of those admits both second and third moments. We can therefore use the Berry-Esseen inequality to write
	\begin{align*}\Prob_0\Big[Y(t)\geq N-1-k\Big]&=\Prob_0\pa{\frac{Y(t)}{\sqrt{tN^2}}\geq \frac{N-1-k}{\sqrt{tN^2}}}\\
	&=1-{\mathcal N}\pa{\frac{N-1-k}{\sqrt{tN^2}}}+O\pa{\frac{1}{Nt^{1/2}}},\end{align*}	where ${\mathcal N}(u)$ is the distribution function of a standard Gaussian variable.
	We can therefore  also write
	\begin{align*}\Prob_0\Big[Y(t-s)\geq N-1-k\Big]=1-{\mathcal N}\pa{\frac{N-1-k}{\sqrt{(t-s)N^2}}}+O\pa{\frac{1}{N\sqrt{t-s}}}.\end{align*}
	These two identities allow us to write, since $p+1\leq k\leq N-1$ 
	\begin{align*}
	\Prob_k(  H^Y(\{N-1\})\in &[t-s,t])\\
	&={\mathcal N}\pa{\frac{N-1-k}{\sqrt{(t-s)N^2}}}-{\mathcal N}\pa{\frac{N-1-k}{\sqrt{tN^2}}}+O\pa{\frac{1}{N\sqrt{t-s}}}\\
	&\leq\frac{1}{\sqrt{2\pi}}\pa{\frac{N-1-k}{\sqrt{(t-s)N^2}}-\frac{N-1-k}{\sqrt{tN^2}}}+O\pa{\frac{1}{N\sqrt{t-s}}}\\
	&\leq\frac{1}{\sqrt{2\pi t}}\pa{\frac{\sqrt{t}}{\sqrt{(t-s)}}-1}+O\pa{\frac{1}{N\sqrt{t-s}}}	
	\end{align*}
One easily obtains after elementary computations a universal constant $C$ such that the first term in the right hand side above is less than $Cs/t^{3/2}$ thus concluding the proof of Lemma \ref{lem:BoundaryProba1}.
\end{proof}

We now use this technical Lemma to prove the following result, which will be needed to estimate the correlations function $\varphi_N$.

\begin{lemma}
	\label{lem:grad} There exists $\varepsilon_0>0$ such that, for any $0<\varepsilon<\varepsilon_0$,  there exists a constant $M$  independent of $N$, such that
	\[\sup_{\substack{t\in[N^{-\varepsilon},T]\\
			k\in \{p+1,\dots,N-1\}}}\abs{\rho_N(t,k+1)-\rho_N(t,k)}\leq M N^{-\frac 12 -\varepsilon}\]
\end{lemma}
\begin{proof}[Proof of Lemma \ref{lem:grad}]This Lemma is a consequence of Proposition \ref{lem:convrhop} and Lemma \ref{lem:BoundaryProba1}. For any $k\in \{ p+1,\dots,N-1 \}$, let us denote 	
	\begin{equation}\label{eq:Defgh}g(t,k)=\rho_N(t,k+1)-\rho_N(t,k),\quad\mbox{ and  }\quad h(k)=g(0,k)=\rho_0(k+1/N)-\rho_0(k/N).\end{equation}
	Using equation \eqref{eq:eqrhoapprox}, we obtain that $g$ is solution to 
	\begin{equation}
	\label{eq:systgrad}
	\begin{cases}
	\partial_tg(t,k)=N^2(\Delta_Ng)(t,k) & \forall k\in \{ p+2,\dots,N-2 \}\\
	g(t,p+1)=\rho_N(t,p+2)-\rho_N(t,p+1)\\
	g(t,N-1)=\beta-\rho_N(t,N-1)\\
	g(0,.)=h(.)
	\end{cases}
	\end{equation}
Recall that we denoted by $Y$ a random walk on $\Z$, and that $H^Y(B)$ is the first time $Y$ hits the set $B\subset \Z$, and let $H^Y_t(B)=H^Y(B)\wedge t$. To keep notations simple, shorten 
\[H_t=H^Y_t(\{p+1,N-1\}).\]
Then, since $g$ satisfies \eqref{eq:systgrad}, we can write for any $(t,k)\in[0,T]\times\{ p+2,\dots,N-2 \} $  
	\begin{equation}
\label{eq:idg}
	g(t,k)=\E_k\Big[g(t-H_t,Y(H_t))\Big].
	\end{equation}
	According to Proposition \ref{lem:convrhop}, the more $H_t$ is close to $t$, the less control we have over the value of $g$ at the boundaries $p+1$ and $N-1$. However, the probability that $Y$ reaches either spatial boundary very close to time $t$ is small according to Lemma \ref{lem:BoundaryProba1}. To make this argument rigorous, we fix a small $\delta>0$, and let 
	\[a_N=N^{-\frac 12-\delta} \quad \mbox{and}\quad b_N=N^{-\frac 12+\delta}\]

	Fix $\delta'>\delta.$ We now consider four cases, represented in Figure \ref{fig:RWgrad}:
	\begin{figure}
		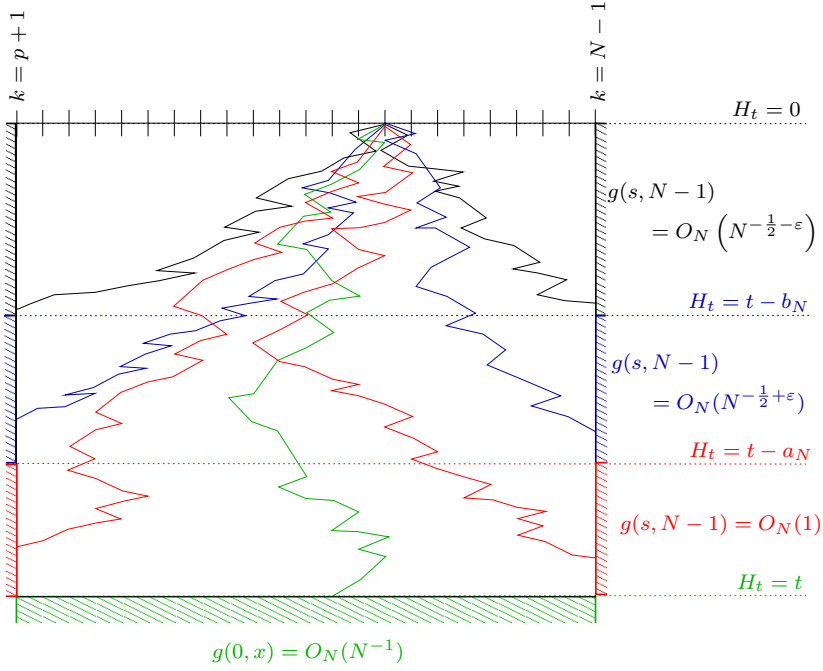
		\caption{Representation of the four possible cases for the random walk $Y$.}
		\label{fig:RWgrad}
	\end{figure}
	\begin{itemize}
		\item If $0\leq H_t\leq t-b_N$, then the random walk has hit the black boundary. Furthermore, thanks to Proposition \ref{lem:convrhop}, we have a good control of the value of $g$ at the boundary, which is of order $N^{-\frac 12 -\delta'}$. 
		\item If $t-b_N\leq H_t\leq t-a_N$, the random walk hits the blue boundary, where we have some control of the value of $g$ thanks to Proposition \ref{lem:convrhop}. We also have control over the probability that $Y$ hits the blue boundary thanks to Lemma  \ref{lem:BoundaryProba1}, so that the overall contribution of this term is of order $N^{-1 +c}$ for some small constant $c$.
		\item If $ t-a_N\leq H_t< t$, we have no good control over the value of $g$ at the boundary red, which is a priori of order $1$. However the probability that $Y$ hits the red boundary is well controlled by Lemma \ref{lem:BoundaryProba1}.
			\item Finally, if $H_t=t$, the random walk reaches the green boundary (i.e. time $0$ for $g$), and we can write \[g(t-H_t,Y(H_t))=h(Y(H_t))=O(1/N).\]
	\end{itemize}
	More precisely, fix a small $\varepsilon>0$, we can write thanks to equation \eqref{eq:idg} for any $t\in [N^{-\varepsilon}, t]$ and any $\delta>0$
	\begin{align}\label{eq:decompcase}
	\abs{g(t,k)}=&\abs{\E_k\pa{g(t-H_t,Y(H_t))}}\nonumber\\
	\leq& \;\Prob_k(0\leq H_t\leq t-b_N)\sup_{s\in [b_N,t[}\abs{g(s,p+1)}\vee\abs{g(s,N-1)}\nonumber\\
	&+ \Prob_k( t-b_N\leq H_t\leq t-a_N)\sup_{s\in \cro{a_N,b_N}}\abs{g(s,p+1)}\vee\abs{g(s,N-1)}\nonumber\\
	&+ \Prob_k( t-a_N\leq H_t< t)\sup_{s\in \cro{0,a_N}}\abs{g(s,p+1)}\vee\abs{g(s,N-1)}\nonumber\\
	&+\Prob_k(H_t=t)\sup_{p+2\leq k\leq N-2}\abs{h(k)}
	\end{align}
	
	\medskip
	
	We now estimate each of these terms : according to Proposition \ref{lem:convrhop}, and by definition of the function $g$, we chose since $\delta'-\delta>0$, we can write for the first term
	\begin{multline}
	\label{eq:firstcase}
	\Prob_k(0\leq H_t\leq t-b_N)\sup_{s\in [b_N,t[}\abs{g(s,p+1)}\vee\abs{g(s,N-1)}\\
	\leq\sup_{s\in [b_N,t[}\abs{g(s,p+1)}\vee\abs{g(s,N-1)}\leq K \frac{N^{\delta'-\delta-1}}{b_N}\leq KN^{\delta'-2\delta-\frac 12}.
	\end{multline}

	\medskip
	
	Regarding the second term, we use this time both Proposition \ref{lem:convrhop} and Lemma \ref{lem:BoundaryProba1}. For any $t\in [N^{-\varepsilon},T] $, we let $t'=t-a_N\geq N^{-\varepsilon}-N^{-\frac 12 -\delta} $ and $s=b_N-a_N\leq 2N^{-\frac 12 +\delta}$, to obtain
		\begin{multline}
	\label{eq:secondcase}
	\Prob_k(t-b_N\leq H_t\leq t-a_N)\sup_{s\in [a_N,b_N[}\abs{g(s,p+1)}\vee\abs{g(s,N-1)}\\
	\leq C\pa{\frac{s}{t'^{3/2}}+\frac{1}{Nt'^{1/2}}}K \frac{N^{\delta'-\delta-1}}{a_N}\leq M_1 N^{\frac 32 \varepsilon +\delta'+\delta-1},
	\end{multline}
	for some constant $M_1$ depending on $C$ and $K$.
	
	\medskip
	
The third term is controlled by Lemma \ref{lem:BoundaryProba1}, and this time we fix $t\in [N^{-\varepsilon},T]$, and let $s=a_N$, to obtain
	\begin{multline}
	\label{eq:thirdcase}
	\Prob_k(t-a_N\leq H_t\leq t)\sup_{s\in [0,a_N[}\abs{g(s,p+1)}\vee\abs{g(s,N-1)}\\
	\leq\Prob_k(t-a_N\leq H_t\leq t)\leq  C\pa{\frac{s}{t^{3/2}}+\frac{1}{Nt^{1/2}}}\leq  M_2 N^{\frac 32 \varepsilon-\frac 12-\delta}.
	\end{multline}
	for some constant $M_2$ depending on $C$.
	
	\medskip
	
	Finally, since $\rho_0$ is was assumed smooth, we also have 
		\begin{align}
	\label{eq:fourthcase}
	\Prob_k(H_t=t)\sup_{p+2\leq k\leq N-2}\abs{h(k)}\leq \frac{\norm{\partial_u\rho_0}{\infty}}{N}.
		\end{align}

	We can now choose 	\[\delta=5\varepsilon/2,\quad \delta'=4\varepsilon>\delta\mbox{ and }\varepsilon_0=1/18 \]
	and inject the four bounds \eqref{eq:firstcase}, \eqref{eq:secondcase}, \eqref{eq:thirdcase} and \eqref{eq:fourthcase} in equation \eqref{eq:decompcase}, to finally obtain that for any $0<\varepsilon<\varepsilon_0$, and any $(t,k)\in[N^{-\varepsilon},T]\times\{ p+2,\dots,N-2 \} $ 
	\begin{equation*}
	\abs{g(t,k)}\leq(K+M_1+M_2+\norm{\partial_u\rho_0}_{\infty})N^{-\frac 12 -\varepsilon}.
	\end{equation*} 
	Letting $M=K+M_1+M_2+\norm{\partial_u\rho_0}_{\infty}$ then completes the proof of Lemma \ref{lem:grad}.
\end{proof} 
The previous estimate yields control over the gradient for macroscopic times of order $N^{\varepsilon}$, uniformly in $\Lambda_N$. We now estimate the gradient of the density for times very close to $0$. Since the initial density is not necessarily close to $\alpha$ at the left boundary, and to $\beta$ at the right boundary, the gradient of the density can be very steep at the boundaries close to the initial time. Away from the boundaries, however, for very small times, the discrete gradient of the density is very close to that of the initial density profile $\rho_0$, and is therefore of order $1/N$. We now make this statement rigorous.
\begin{lemma}
	\label{lem:grad2}
	 Let us denote $x_{N,\varepsilon}=N^{1-\varepsilon/4}$. For any $\varepsilon>0$, there exists a constant $M'=M'(\varepsilon, \norm{\partial_u\rho_0}_{\infty})$ such that
	\[\sup_{\substack{t\in[0,N^{-\varepsilon}]\\
			k\in \{x_{N,\varepsilon},\dots,N-x_{N,\varepsilon}\}}}\abs{\rho_N(t,k+1)-\rho_N(t,k)}\leq \frac{M'}{N}.\]
\end{lemma}
\begin{proof}
	The proof of this statement also comes from duality. This time, however, the random walk $Y$ is started at a distance at least  $x_{N,\varepsilon}=N^{1-\varepsilon/4}$ from the boundary, so that the probability that in a macroscopic time of smaller than $N^{-\varepsilon}$ (i.e. in a microscopic time of order $N^{2-\varepsilon}$), it travels such a distance vanishes exponentially in $N^{\varepsilon/4}$.
	Once again, we shorten $H_t=H_t^Y(\{p+1, N-1\}).$ Recall from equation \eqref{eq:Defgh} the definitions of $g$ and $h$. Since $\abs{g}\leq 1$, following the same steps and using the same notations as in the previous Lemma, we can write for any $k\in \Lambda_N$ 
	\begin{align}
	\abs{g(t,k)}&\leq 2\Prob_k(H_t<t)+\Prob_k(H_t=t) \sup_{p+2\leq k\leq N-2}\abs{h(k)}\\
	& \leq 2\Prob_k(H_t<t) +\frac{\norm{\partial_u\rho_0}_{\infty}}{N}
	.\end{align}
	As mentioned before, for any $k\in \{x_{N,\varepsilon},\dots,N-x_{N,\varepsilon}\}$, $\Prob_k(H_t<t)$ is less than the probability that a rate $1$ symmetric random walk travels in a time $\delta t=N^{2-\varepsilon}$ a distance $\delta x=N^{1-\varepsilon/4}=N^{\varepsilon/4}\sqrt{\delta t  }$, which can be bounded by $e^{-C N^{\varepsilon/4}}$ for some positive constant $C$ depending only on $\varepsilon$ but not on the starting point $k\in \{x_{N,\varepsilon},\dots,N-x_{N,\varepsilon}\}$, which proves the Lemma.	
\end{proof}

\begin{corollary}
	\label{cor:grad}
	 There exists $\varepsilon>0$ and a constant $M_0$  independent of $N$, such that
	\[\sup_{\substack{t\in[0,T]\\
			k\in \{x_{N,\varepsilon},\dots,N-x_{N,\varepsilon}\}}}\Big\{\rho_N(t,k+1)-\rho_N(t,k)\Big\}^2\leq M_0 N^{-1 -2\varepsilon}\]
\end{corollary}
This corollary is an immediate consequence of Lemmas \ref{lem:grad} and \ref{lem:grad2}, by choosing any $\varepsilon \leq \varepsilon_0$.

\section{Estimation of the correlation function}
\label{sec:phi}
\subsection{Notations}
\label{sec:Correlations1} We now use Lemmas \ref{lem:grad} and \ref{lem:grad2} to estimate the correlations of the model. The estimation is stated in Proposition \ref{prop:BoundCor}, and uses similar tools as in the previous sections : we obtain a discrete differential system satisfied by the correlation function $\varphi_N$, and use duality to estimate $\varphi_N$ using two-dimensional random walk. 

\medskip

Recall from equation \eqref{eq:DefCorHydro} that we defined the correlation function 
\begin{equation}
\varphi_N(t,k,l)=\E_{\nu_N}\Big(\{\eta_k(t)-\rho_N(t,k)\}\{\eta_l(t)-\rho_N(t,l)\}\Big).
\end{equation}
We will denote the two-dimensional equivalents of one-dimensional devices by bolds characters.
In particular, we denote pairs of integers by $\bfk=(k, l)\in \Z^2$. For any $\bfk=(k, l)\in \Z^2$, let 
\[\abs{\bfk}=\abs{k-l},\quad \mbox{ and }\quad \norm{\bfk}=\abs{k}\vee\abs{l}.\]
For any $\bfk=(k,l)$, we denote $\boldsymbol \Delta_N$ the two-dimensional discrete Laplacian
\begin{align}
\label{eq:lap2D}
(\boldsymbol \Delta_N \varphi_N)(\bfk)=&\sum_{\bfk'\sim\bfk}(\varphi_N(\bfk')-\varphi_N(\bfk))\nonumber\\
=&\;\varphi_N(k+1, l)+\varphi_N(k-1, l)+\varphi_N(k, l+1)+\varphi_N(k, l-1)\nonumber \\
&-4\varphi_N(k, l),
\end{align}
and by $\boldsymbol \nabla_N\varphi$ the diagonal "gradient"
\begin{equation}
\label{eq:gradext}
(\boldsymbol \nabla_N\varphi)(\bfk)=\varphi_N(k-1, l)+\varphi_N(k, l+1)-2\varphi_N(k, l).
\end{equation}
For the convenience of notations, we will sometimes write $N^{3/4}$ instead of $\lfloor N^{3/4} \rfloor.$
\begin{figure}
	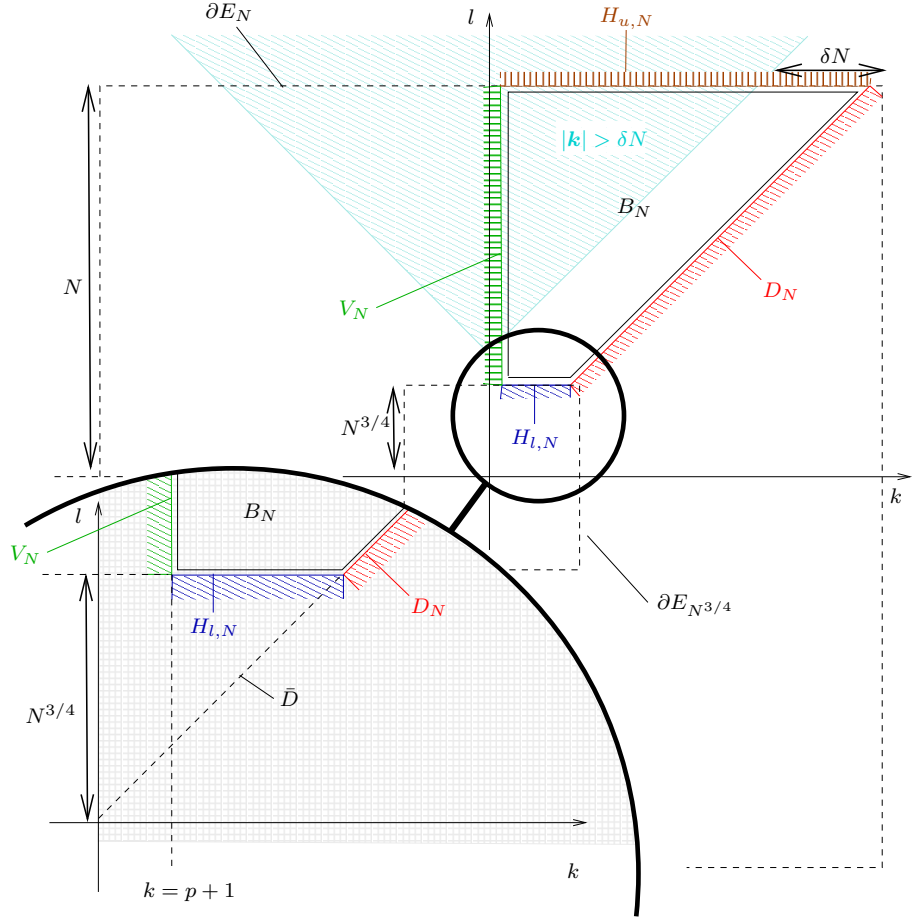
	\caption{Representation of the bulk $B_N$ (black), the diagonal border $D_N$ (red), the vertical border $V_N$ (green),  the lower horizontal border $H_{l,N}$ (blue) and  the upper horizontal border $H_{u,N}$ (brown). }
	\label{fig:SDBL}
\end{figure}

\medskip

As represented in Figure \ref{fig:SDBL}, let us introduce the \emph{bulk} 
\[B_N=\left\{(k,l),\quad p+1< k<N^{3/4}\vee (l-1), \quad N^{3/4} < l <N \right\},\]
the \emph{diagonal boundary}
\[D_N=\left\{(k,k+1), \quad N^{3/4}\leq k\leq N-2\right\},\]
the left \emph{vertical boundary}
\[V_N=\left\{(p+1,l), \quad N^{3/4}< l< N\right\},\]
 the \emph{lower horizontal border}
\[H_{l,N}=\left\{(k,N^{3/4}), \quad p+1\leq k< N^{3/4}\right\},\]
and the \emph{upper horizontal border} 
\[H_{u,N}=\left\{(k,N), \quad p+1\leq k\leq N-2\right\}.\]
Finally, we denote $\partial B_N=V_N\cup H_{l,N}\cup H_{u,N}$. 

Our main result is the following, and states that for any positive $\delta$, at a distance of order $\delta N$ of both extremities of the diagonal, the correlations vanish uniformly as $N$ goes to $\infty$. Let us finally shorten $\0=(0,0)$ and $\bfN=(N,N)$.
\begin{proposition}
	\label{prop:BoundCor}
	For any $t\in [0,T]$, and any $\delta>0$
	\begin{equation*}
	\limsup_{N\to\infty}\sup_{\substack{\bfk\in B_N\\
			\norm{\bfk}, \norm{\bfk-\bfN}>\delta N}}\abs{\varphi_N(t,\bfk)}=0.
	\end{equation*}
\end{proposition}
\begin{proof}[Proof of Proposition \ref{prop:BoundCor}]
For any $(t,\bfk)\in [0,T]\times\BFL$
\[\partial_t\varphi_N(t,\bfk)=N^2\E_{\nu_N}\Big[L_N\{\eta_k(t)-\rho_N(t,k)\}\{\eta_l(t)-\rho_N(t,l)\}\Big].\]
We will use the notation $\eta_N(t)=\beta=\rho_N(N)$, so that we can extend the definition of $\varphi_N$ for any time $t$ and any $\bfk=(k,N)$ in the upper boundary $H_{u,N}$, and let 
\[\varphi_N(t,\bfk)=0.\]
With this notation, which defines $\varphi_N$ at the upper boundary $H_{u,N}$, elementary computations then yield for any $t\in [0,T]$, and  any $\bfk$ in $B_N$ 
\[\partial_t\varphi_N(t,\bfk)=N^2(\boldsymbol\Delta_N\varphi_N)(t,\bfk)\]
where $\boldsymbol\Delta_N$ is the discrete two-dimensional Laplacian introduced earlier.
We obtain in the same way, for any $\bfk\in D_N$, that  
\[ \partial_t\varphi_N(t,\bfk)=N^2(\boldsymbol\nabla_N\varphi_N)(t,\bfk))-N^2m(t,\bfk),\]
where we denoted for $\bfk=(k,k+1)\in D_N$
\[m(t,\bfk)=\big\{\rho_N(t,k+1)-\rho_N(t,k)\big\}^2,\]
and $\boldsymbol\nabla_N$ is the gradient introduced in equation \eqref{eq:gradext}, representing reflection at the diagonal.
We do not know yet the value of $\varphi_N$ neither on the vertical boundary nor on the lower horizontal boundary $H_{l,N}$. However, we already obtained  the behavior at the diagonal boundary $D_N$.

\medskip

We started our process from a product measure, so that there are no correlations at time $0$. By the previous statements, the correlation function $\varphi_N$ is therefore solution to the discrete difference system
\begin{equation}
\label{eq:eqcor1}
\begin{cases}
\partial_t\phi(t,\bfk)=N^2(\boldsymbol\Delta_N\phi)(t,\bfk) &\forall (t, \bfk) \in [0,T]\times B_N\\
\partial_t\phi(t,\bfk)=N^2(\boldsymbol\nabla_N\phi)(t,\bfk)-N^2m(t, \bfk) &\forall (t, \bfk) \in [0,T]\times D_N\\
\phi(t, \bfk)=\varphi_N(t, \bfk) & \forall (t, \bfk) \in [0,T]\times (V_N\cup H_{l,N} )\\
\phi(t, \bfk)=0 &\forall (t, \bfk) \in [0,T]\times H_{u,N}\\
\phi(0, \bfk)=0&\forall \bfk\in B_N
\end{cases}.
\end{equation}
Note in particular that the third line gives no informations, but we include it in order to write a discrete difference system with complete boundary conditions.
Like we did for the density, we are going to pair $\varphi_N$ with a random walk $\bfX$.

\subsection{Pairing with a random walk}

We introduce the infinite diagonal
\[\bar D:=\{(k,k+1),\quad  k\in \Z\}.\]
Note in particular that $D_N\subset \bar D$. We denote by $\bfX$ a random walk on $\Z^2$ driven by the generator $N^2\gene$, where for any function $f:\Z^2\to {\bb R}$
\begin{equation}\label{eq:genebfX}(\gene f)({\bf x})=\ind{\bfx\notin \bar D}(\boldsymbol \Delta_Nf)(\bfx)+\ind{\bfx\in \bar D}(\boldsymbol \nabla_Nf)(\bfx).\end{equation}
In other words, $X$ performs a symmetric random walk in $\Z^2$, and is reflected when hitting $\bar D$.
We also denote by $\Prob_{\bfk}$ the distribution of this random walk, started from $\bfk$, and by $\E_{\bfk}$ the corresponding expectation. Similarly to the one-dimensional notations, for any set $S$, we  denote by $\bfH(S)$ the hitting time of $S$ and let $\bfH_t(S)=\bfH(S)\wedge t$. 
By duality, analogously to the previous section, since $\varphi_N$ is solution of \eqref{eq:eqcor1}, we can then write for any $(t,\bfk)\in[0,T]\times B_N$
\begin{multline}
\label{eq:phiX}
\varphi_N(t,\bfk)=\E_{\bfk}\Bigg[\varphi_N\Big(t-\bfH_t(\partial B_N), \bfX(\bfH_t(\partial B_N))\Big)\\
-N^2\int_{s=0}^{\bfH_t(\partial B_N)}\ind{\bfX(s)\in D_N}m(t-s,\bfX(s))ds\Bigg].
\end{multline}
Let us denote 
\begin{equation}
\label{eq:DefcN} 
c_{N}=\sup_{\substack{t\in [0,T]\\
		\bfk\in V_N}}\abs{\varphi_N(t,\bfk)},
\end{equation}
and note that $\abs{\varphi_N(t,\bfk)}\leq 1$ for any $t$ and any $\bfk$.
Equation \eqref{eq:phiX} yields that for any $(t,\bfk)\in[0,T]\times B_N$
\begin{equation}
\label{eq:phiXpsi}
\abs{\varphi_N(t,\bfk)}\leq\psi_N(t,\bfk)+N^2\E_{\bfk}\pa{\int_{s=0}^{\bfH_t(\partial B_N)}\ind{\bfX(s)\in D_N}m(t-s,\bfX(s))ds},
\end{equation}
where $\psi_N$ is solution to the system
\begin{equation}
\label{eq:Systempsi}
\begin{cases}
\partial_t\phi(t,\bfk)=N^2(\boldsymbol\Delta_N\phi)(t,\bfk) &\forall (t, \bfk) \in [0,T]\times B_N\\
\partial_t\phi(t,\bfk)=N^2(\boldsymbol\nabla_N\phi)(t,\bfk) &\forall (t, \bfk) \in [0,T]\times D_N\\
\phi(t, \bfk)=c_N & \forall (t, \bfk) \in [0,T]\times V_N\\
\phi(t, \bfk)=1 & \forall (t, \bfk) \in [0,T]\times H_{l,N}\\
\phi(t, \bfk)=0 &\forall (t, \bfk) \in [0,T]\times H_{u,N}\\
\phi(0, \bfk)=0&\forall \bfk\in B_N
\end{cases}.
\end{equation}
The only difference with 
\eqref{eq:eqcor1} is that we dropped the diagonal increment $m$, and crudely bounded $\varphi_N$ by $c_N$ on $V_N$ and  by $1$ on $H_{l, N}$.

\medskip

In Corollary \ref{cor:grad}, we obtained control over the value of the increment $m(t,k, k+1)$ for $t\in [0,T]$ and $k\in \{N^{1-\varepsilon/4},\dots,N-N^{1-\varepsilon/4}\}.$ However, close to the extremities of $D_N$, $m$ is a priori of order $1$, which is an issue due to the factor $N^2$ in front the increment term. We are therefore  going to kill the random walk $\bfX$ when it gets close to either one of the extremities of $D_N$ before $\bfH(\partial B_N)$, and prove that the difference made by doing so is small.

\medskip

Fix $\varepsilon>0$ given by Corollary \ref{cor:grad}, we define 
\begin{equation}
\label{eq:DefDNe}
\widetilde{D}_{N,\varepsilon}=\{(k,k+1)\in D_N, \quad k\leq N^{1-\varepsilon/4}\mbox{ or }k\geq N-N^{1-\varepsilon/4}\},
\end{equation}
which is the part of the diagonal $D_N$ where we do not have sufficient control over the diagonal increment $m$.  Shorten
 \begin{equation}
  \label{eq:Deftautilde} 
\bfH_t^{\varepsilon}=\bfH_t(\partial B_N\cup\widetilde{D}_{N,\varepsilon}),
\end{equation}
By killing the random walk $\bfX$ at the extremities of the diagonal, we make sure that it does not spend time in the part of $D_N$ where the function $m$ is not well controlled.
We can write for this new stopping time
\begin{equation}
\label{eq:phiX2}
\varphi_N(t,\bfk)=\E_{\bfk}\pa{\varphi_N(t-\bfH_t^{\varepsilon}, \bfX(\bfH_t^{\varepsilon}))-N^2\int_{s=0}^{\bfH_t^{\varepsilon}}\ind{\bfX(s)\in D_N}m(t-s,\bfX(s))ds}.
\end{equation}
We already pointed out that for any $(t, \bfk)\in [0,T]\times \Lambda_N$,  
\[\E_{\bfk}\Big[\big|\varphi_N(t-\bfH_t(\partial B_N), \bfX(\bfH_t(\partial B_N)))\big|\Big]\leq\psi_N(t,\bfk),\]
where $\psi_N$ is the solution to \eqref{eq:Systempsi}. Since $\varphi_N$ is bounded in absolute value by $1$, the bound above yields
\[\E_{\bfk}\pa{\abs{\varphi_N(t-\bfH_t^{\varepsilon}, \bfX(\bfH_t^{\varepsilon}))}}\leq\psi_N(t,\bfk)+2\Prob_{\bfk}\Big[\bfH(\partial B_N)>\bfH(\widetilde{D}_{N, \varepsilon})\Big].\]
Thanks to \eqref{eq:phiX2}, we can therefore write 
\begin{multline*}
\abs{\varphi_N(t,\bfk)}\leq\psi_N(t,\bfk)+2\Prob_{\bfk}\Big[\bfH(\partial B_N)>\bfH(\widetilde{D}_{N, \varepsilon})\Big]\\
+N^2\E_{\bfk}\pa{\int_{s=0}^{\bfH_t^{\varepsilon}}\ind{\bfX(s)\in D_N}m(t-s,\bfX(s))ds}.
\end{multline*}
Proposition \ref{prop:BoundCor}  follows from this estimate and Lemmas \ref{lem:firstterm}, \ref{lem:Secondterm} and \ref{lem:thirdterm}  below.
\end{proof}

\begin{lemma}
	\label{lem:firstterm}
	For any $t\in [0,T]$,  and any $\delta>0$
	\[\limsup_{N\to\infty}\sup_{\substack{\bfk\in B_N\\
	\norm{\bfk}>\delta N}}\psi_N(t,\bfk)=0.\]
\end{lemma}
\begin{lemma}
	\label{lem:Secondterm}
For $\varepsilon>0$ given by Corollary \ref{cor:grad}, and for any $t\in [0,T]$,
	\begin{equation}
	\label{eq:timespentDN}
	\limsup_{N\to\infty}\sup_{\bfk\in B_N}N^2\E_{\bfk}\pa{\int_{s=0}^{\bfH_t^{\varepsilon}}\ind{\bfX(s)\in D_N}m(t-s,\bfX(s))ds}=0.
	\end{equation}
\end{lemma}

\begin{lemma}
	\label{lem:thirdterm}
	For any $\delta>0$, and any $\varepsilon>0$
	\[\limsup_{N\to\infty}\sup_{\substack{(t,\bfk)\in [0,T]\times B_N\\\norm{\bfk}, \norm{\bfk-\bfN}> \delta N}}\Prob_{\bfk}\Big[\bfH(\partial B_N)>\bfH(\widetilde{D}_{N, \varepsilon})\Big]=0.\]
\end{lemma}

For the sake of clarity, we prove these three results in separate sections, before completing the proof of Theorem \ref{thm:HydroGen}. To prove these Lemmas, however, the reflected boundary condition at $D_N$ is not convenient. To solve this issue, recall that we defined 
\[\bar D:=\{(k,k+1),\quad  k\in \Z\}\supset D_N,\]
we now introduce the symmetry operator $\sigma:\Z^2 \to\Z^2$ w.r.t. $\bar D$,
\[\sigma(k,l)=(l-1,k+1).\]
 We are going to make all the items already introduced symmetric w.r.t. $\bar D$. For any set $S\subset \Z^2$ denote 
 \[S^{\sigma}=S\cup\sigma S,\]
 and for any function $f$ defined on some subset $S\subset\{(k,l)\in \Z^2,\quad  k<l\}$ of the half plane above the line $D_N$, we extend it as a function $f^\sigma$ on $S^{\sigma}$ by symmetry, by letting for any $\bfk\in S$
\[f^\sigma(\sigma \bfk)=f(\bfk).\]
	For any $\bfk\in B_N$, we denote by $\bfX^\sigma$ a random walk on $\Z^2$, started from $\bfk$ and driven by the generator $N^2\gene^\sigma$, where for any function $f:\Z^2\to {\bb R}$
\begin{equation}\label{eq:genebfXsigma}(\gene^\sigma f)({\bf x})=\ind{\bfx\notin \bar D}(\boldsymbol \Delta_Nf)(\bfx)+\frac{1}{2}\ind{\bfx\in \bar D}(\boldsymbol \Delta_Nf)(\bfx).\end{equation}

We will denote with exponents $\sigma$ all the corresponding quantities relative to $\sigma$. Note in particular that $\bfX^\sigma$ is no longer reflected at $\sigma$, but it is rather reflected at rate 1/2 and crosses $\bar D$ at rate $1/2$.
With the exception of the time spent on $\bar D$, which is double the time spent in any other place, $\bfX^\sigma$ thus behaves like a rate $N^2$ continuous time random walk on $\Z^2$. 
We denote $\bfH^\sigma(S)$ the hitting time of the symmetrized set $S^\sigma$ by $ \bfX^\sigma$, and once again $\bfH_t^\sigma(S)=\bfH^\sigma(S)\wedge t$. 
The boundary $\partial B_N^\sigma=V_N^\sigma\cup  H_{l,N}^\sigma\cup H_{u,N}^\sigma$ is  represented in Figure \ref{fig:Exitsym}. 
Further note that we can couple $\bfX$ and $\bfX^\sigma$ in a way that for any set $S$ contained in the half plane above $\bar D$, 
\begin{equation}
\label{eq:tausigma2}
\bfH(S)=\bfH^\sigma(S^{\sigma}).
\end{equation}
To build this coupling, given $\bfX$, one simply has to replace with probability $1/2$, independently, each excursion performed by $\bfX$ away from $\bar D$ by its image by the symmetry $\sigma$.

We will always assume in what follows that $\bfX $ and $\bfX^\sigma$ are defined under that coupling, and not to burden the notations, still denote $\Prob_\bfk$ the corresponding distribution.

\subsection{Proof of Lemma \ref{lem:firstterm}}Before estimating the function $\psi_N$, we start estimating the correlations between sites $p+1$ and $k> N^{3/4}$ to obtain an upper bound on the quantity $c_N$ defined in \eqref{eq:DefcN}.
	\begin{lemma}
		\label{lem:cN}
		\[\limsup_{N\to\infty}c_N=\limsup_{N\to\infty}\sup_{\substack{t\in [0,T]\\
				\bfk\in V_N}}\abs{\varphi_N(t,\bfk)}=0.\]  
	\end{lemma}
Before proving this Lemma, we show that it implies Lemma \ref{lem:firstterm}.
Since $\psi_N$ is solution to \eqref{eq:Systempsi}, we can write for any $(t,\bfk)\in [0,T]\times B_N$ 
\[\psi_N(t,\bfk)=\E_{\bfk}\Big[\psi_N\big(t-\bfH_t(\partial B_N),\bfX(\bfH_t(\partial B_N))\big)\Big].\]
Thanks to Lemma \ref{lem:cN}, $\psi_N$ vanishes uniformly in space and time at the vertical boundary $V_N$. Furthermore, $\psi_N$ also vanishes at time $0$ and at the upper boundary $H_{u,N}$. Therefore the only boundary where $\psi_N$ does not ultimately vanish is $H_{l,N}$. Using the coupling between $\bfX$ and $\bfX^\sigma$, and the symmetry  identity \eqref{eq:tausigma2}, we can write for any $t$ and any $\bfk\in B_N$
\begin{align}
\label{eq:estpsitk}
\psi_N(t,\bfk)\leq&\;\Prob_{\bfk}\Big[\bfH_t(\partial B_N)=\bfH(H_{l,N})\Big]+c_N\Prob_{\bfk}\Big[\bfH_t(\partial B_N)=\bfH(V_N)\Big]\nonumber\\
\leq &\;\Prob_{\bfk}\Big[\bfH^\sigma_t(\partial B_N^\sigma)=\bfH^\sigma(H^\sigma_{l,N})\Big]+c_N\nonumber\\
\leq &\; \Prob_{\bfk}\Big[\bfH^\sigma(H^\sigma_{l,N})<\bfH^\sigma(V^\sigma_N)\vee\bfH^\sigma(H^\sigma_{u,N})\Big]+c_N\nonumber\\
\leq &\;\Prob_{\bfk}\Big[\bfH^\sigma(\partial E_{N^{7/8}})<\bfH^\sigma(\partial E_{2N})\Big]+c_N,
\end{align}
where for any integer $K$, $\partial E_K$ is the boundary of the box of side $2K$, centered at $\0$ 
\[\partial E_K=\{\bfk\in \Z^2,\quad  \norm{\bfk}=K\}.\]
The last bound is justified in Figure \ref{fig:Exitsym}, where it is shown that if $\bfX^{\sigma}$ starts from $B_N$ and if $\bfH(H^\sigma_{l,N})<\bfH(V^\sigma_N)\vee\bfH(H^\sigma_{u,N})$, then $\bfX^{\sigma}$ reaches $\partial E_{N^{7/8}}$ before $\partial E_{2N}$.

\begin{figure}
	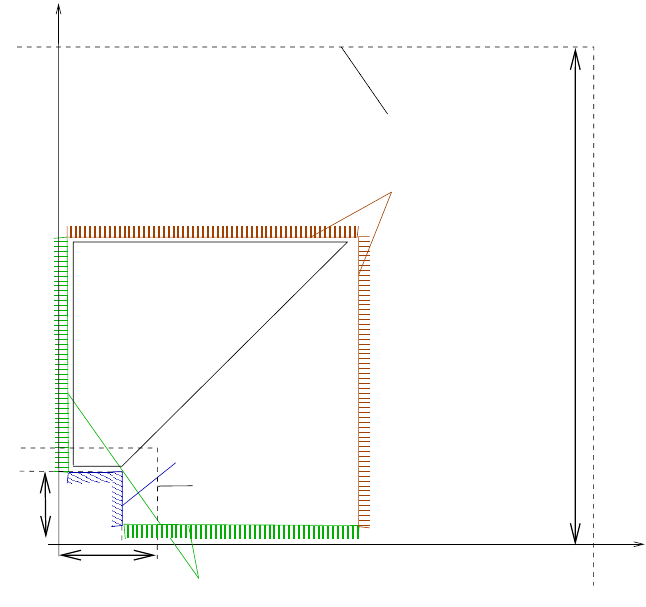
	\caption{Representation of the sets $H_{l,N}^{\sigma} $, $V_N^\sigma$, and $H_{u,N}^\sigma$. Starting from $ B_N \cap E^c_{N^{7/8}}$, in order to leave the area delimited by the hatched boundary at  $H_{l,N}^{\sigma} $, a random walk must hit $\partial E_{N^{7/8}}$  before hitting $ H_{l,N}^{\sigma} $. Colour figure online.}
	\label{fig:Exitsym}
\end{figure}

Furthermore, (cf. Exercise 1.6.8 in \cite{Lawler1996})
\begin{equation}
\label{eq:p2N78N}
\sup_{\substack{\bfk\in B_N\\		\norm{\bfk}>\delta_N}}\Prob_{\bfk}(\bfH(\partial E_{N^{7/8}})<\bfH(\partial E_{2N}))\sim\frac{\log 2N-\log{\delta N}}{\log 2N-\log N^{7/8}}=\frac{8\log\delta}{8\log2 +\log N},
\end{equation}
which vanishes as $N\to\infty$ for any fixed $\delta$, as wanted. We now only need to combine equations \eqref{eq:estpsitk} and \eqref{eq:p2N78N}, and Lemma \ref{lem:cN}, to prove Lemma \ref{lem:firstterm}.

\bigskip

We now prove Lemma \ref{lem:cN}.
\begin{proof}[Proof of Lemma \ref{lem:cN}]
In order to prove this Lemma, consider two random walks $X_1$ and $X_2$, respectively started from $p+1$ and $N^{3/4}<k<N$, and both driven by the dual generator $\gend_N$ defined after \eqref{eq:DefLCd}. We are going to prove that before these two particles get close to each other, $X_1$ will have reached one of the cemetery states $\dd_1,\dots,\dd_p$ with high probability. Let us denote 
\[\bar D_p=\{\bfk\in \Z^2, \quad \abs{\bfk}=p+1\}\]
Let $\bfH(\bar D_p)$ be the first time these random walks are at a distance $p+1$, 
\[\bfH(\bar D_p)=\inf\{t\geq0,\quad (X_1(t),X_2(t))\in \bar D_p\}.\]
Let us define $x_{1,N}=N^{3/4}/4-1$ and $x_{2,N}=3N^{3/4}/4+1$
and let $H^1_N, H^2_N$ be defined as 
\[H^i_N=\inf\left\{t\geq 0, \quad X_i(s)=x_{i,N}\right\}.\]
Note in particular that in order for $\bfX:=(X_1(t),X_2(t))\in \bar D_p$ to occur, either $X_1$ must have reached $x_{1,N}$ or $X_2$ must have reached $x_{2,N}$, so that 
\begin{equation}
\label{eq:tau1tau2}
\bfH(\bar D_p)>H_N^1\wedge H_N^2.
\end{equation}
Finally, denote  
\begin{equation*}
\bfH(\dd)= H^{X_1}(\partial \Lambda_N\setminus\{N\}),
\end{equation*}
the hitting time of one of the cemetery states $\dd_k$ by the first coordinate of the random walk $\bfX$.

\bigskip

Let us denote by $\widetilde{\Prob}_{k}$ the joint distribution of $X_1$, $X_2$, moving independently, where the first is started from $p+1$ and the second from $k\geq N^{3/4}$.
We claim that $X_1$ reaches one of the cemetery states before the two random walks get close to each other, or 
\begin{equation}
\label{eq:noencounter}
\limsup_{N\to\infty}\sup_{N^{3/4}\leq k\leq N}\widetilde{\Prob}_{k}\Big[\bfH(\dd)>\bfH(\bar D_p)\Big]=0.
\end{equation}

\bigskip

Denote $s_N=N^{-3/4}$. For any $N^{3/4}\leq k\leq N$, equation \eqref{eq:tau1tau2} yields
\begin{equation}
\label{eq:pX1X2dd}
\widetilde{\Prob}_{k}\Big[\bfH(\dd)>\bfH(\bar D_p)\Big]\leq \widetilde{\Prob}_{k}\Big[\bfH(\dd)\geq s_N\Big]+ \widetilde{\Prob}_{k}\Big[\bfH(\dd)\geq  H^1_N\Big]+\widetilde{\Prob}_{k}\Big[H_N^2\leq s_N\Big].
\end{equation}
Since $X_1$ and $X_2$ behave as random walkers until they are in $\lps$,  the last term is less than the probability that a rate $N^2$ symmetric random walk on $\Z$ travels a distance of order $N^{3/4}$ before time $s_N$. However, because of the acceleration in $N^2$, in a time $s_N$,  $X_2$ would typically travel a distance of order $\sqrt{s_NN^2}=N^{5/8}=N^{3/4}N^{-1/8}.$ Elementary computations and a large deviations estimate therefore yields that for some constant $C$ independent of $N$ and $k\geq N^{3/4}$
\[\widetilde{\Prob}_{k}(H_N^2\leq s_N)\leq e^{-CN^{1/8}}.\]
Furthermore, using minimal adaptations of equations \eqref{eq:boundt} and \eqref{eq:boundN}, we obtain 
\begin{equation*}
\widetilde{\Prob}_{k}(\bfH(\dd)\geq s_N)\leq  \frac{K(\log N)^2}{s_NN},
\end{equation*}
where $K$ is a constant depending on $C$ and $\pi$, and 
\begin{equation*}
\widetilde{\Prob}_{k}(\bfH(\dd)\geq H_N^1)\leq\frac{2}{\pi(x_{1,N}-1-p)}.
\end{equation*}

These three bounds and equation \eqref{eq:pX1X2dd} yield that for any $N^{3/4}\leq k\leq N-1$, $\varepsilon <1/4$, and $N$ large enough, 
\[\widetilde{\Prob}_{k}\Big[\bfH(\dd)>\bfH(\bar D_p)\Big]\leq N^{-\varepsilon},\]
which proves equation \eqref{eq:noencounter}.

\bigskip

We now get back to estimating the correlations and proving Lemma \ref{lem:cN}. In order not to introduce burdensome notations, we will not write in full detail this part of the proof, which relies once again on duality. As we did to estimate the density, we pair $\varphi_N$ with a  two-dimensional random walk $\bfX=(X_1,X_2)$ on $\bar \Lambda_N^2$. Let us shorten
\[\bfH=\bfH(\dd)\wedge \bfH(\bar D_p) \quad \mbox{ and } \quad \bfH_t=\bfH\wedge t .\]
If $\bfX$ reaches either time $t$ or one of the cemetery states $(\dd_k,l)$, with $l>p+1$ not in the boundary, $\varphi_N$ vanishes. Furthermore, since before $\tau$, we have $\abs{X_1- X_2}>p+1$, $X_1$ and $X_2$ are distributed at least until $\tau$ as independent random walks with generator $N^2\gend_N$. We can therefore write, noting that $\bfX$ cannot reach the diagonal $\bar D$ before time $\tau$, and since $\abs{\varphi_N}$ is less than $1$ and vanishes at time $0$,
\begin{multline*}
\varphi_N\pa{t,p+1,k}=\E_{(p+1,k)}(\varphi_N(t-\bfH_t, \bfX(\bfH_t))\\
\leq\widetilde{\Prob}_{k}(\bfH(\dd)>\bfH(\bar D_p))+\E_{(p+1,k)}(\underset{=0}{\underbrace{\varphi_N(t-\bfH, \bfX(\bfH))\ind{\bfH(\dd)\leq\bfH(\bar D_p)}}}),
\end{multline*} 
so that \eqref{eq:noencounter} concludes the proof of Lemma \ref{lem:cN}.
\end{proof}

\subsection{Proof of Lemma \ref{lem:Secondterm}} We now estimate the overall contribution of the diagonal increments to $\varphi_N$. Recall that we want to estimate
	\[N^2\E_{\bfk}\pa{\int_{s=0}^{\bfH_t^{\varepsilon}}\ind{\bfX(s)\in D_N}m(t-s,\bfX(s))ds},\]
	where we shortened $\bfH_t^{\varepsilon}=\bfH_t(\partial B_N\cup \widetilde{D}_{N,\varepsilon})$ and $\widetilde{D}_{N, \varepsilon}$ is the set of points -defined in \eqref{eq:DefDNe}- of the diagonal $D_N$ at distance at most $N^{1-\varepsilon/4}$ of its extremities. 
By definition of $\bfH_t^{\varepsilon}$, for any $s\in [0,\bfH_t^{\varepsilon})$, we cannot have $\bfX\in \widetilde{D}_{N,\varepsilon}$. In particular, for $\varepsilon>0$ given by Corollary \ref{cor:grad}, $\bfk\in B_N$,
	\begin{multline*}N^2\E_{\bfk}\pa{\int_{s=0}^{\bfH_t^{\varepsilon}}\ind{\bfX(s)\in D_N}m(t-s,\bfX(s))ds}\\
	\leq  M_0N^{1-2\varepsilon}\E_{\bfk}\pa{\int_{s=0}^{\bfH_t^{\varepsilon}}\ind{\bfX(s)\in D_N}ds}.\end{multline*}
	Now that the problem of controlling $m$ is dealt with, we can get back to the real stopping time $\bfH_t(\partial B_N)$ (which is by definition larger than $\bfH_t^{\varepsilon}$) and write 
 \begin{multline}\label{eq:TpsDN}N^2\E_{\bfk}\pa{\int_{s=0}^{\bfH_t^{\varepsilon}}\ind{\bfX(s)\in D_N}m(t-s,\bfX(s))ds}\\
 \leq  M_0N^{1-2\varepsilon}\E_{\bfk}\pa{\int_{s=0}^{\bfH_t(\partial B_N)}\ind{\bfX(s)\in D_N}ds}.\end{multline}

	In order to simplify the problem, we start by making it symmetric w.r.t. the line $\bar D=\{(k,k+1),\quad k\in \Z\}$. To do so, we use once again the random walk $\bfX^\sigma$ introduced earlier, with generator given by \eqref{eq:genebfXsigma}. By construction, $\bfX$ and $\bfX^{\sigma}$ spend the same time in $D_N$, therefore
	\[\int_{s=0}^{\bfH_t(\partial B_N)}\ind{\bfX(s)\in D_N}ds=\int_{s=0}^{\bfH^\sigma_t(\partial B_N^\sigma)}\ind{\bfX^{\sigma}(s)\in D_N}ds,\]
where 
$\bfH^\sigma_t(\partial B_N^\sigma)$ was introduced just before \eqref{eq:tausigma2}. 	Recall that $\partial E_{2N}$ is the set of vertices $\bfk$ such that $\norm{\bfk}=2N,$ and that $ \bfH (\partial E_{2N})$ is the first time $\bfX$ hits the boundary $\partial E_{2N}$. Assuming that $\bfX^\sigma$ starts in $B_N$, we can write according to Figure \ref{fig:Exitsym} that $\bfH^\sigma_t(\partial B_N^\sigma)\leq \bfH^\sigma(\partial B_N^\sigma)<\bfH^\sigma(\partial E_{2N})$, so that 
\[\int_{s=0}^{\bfH_t(\partial B_N)}\ind{\bfX(s)\in D_N}ds\leq\int_{s=0}^{\bfH^\sigma(\partial E_{2N})}\ind{\bfX^{\sigma}(s)\in D_N}ds.\]
This last bound and equation \eqref{eq:TpsDN} finally yield that for any $N$ large enough
	\begin{multline*}N^2\E_{\bfk}\pa{\int_{s=0}^{\bfH_t(\partial B_N)}\ind{\bfX(s)\in D_N}m(t-s,\bfX(s))ds}\\
\leq M_0 N^{1-2\varepsilon}\E_{\bfk}\pa{\int_{s=0}^{\bfH^{\sigma}(\partial E_{2N})}\ind{\bfX^{\sigma}(s)\in D_N}ds},\end{multline*}
therefore Lemma \ref{lem:Secondterm} follows from Lemma \ref{lem:secondterm} below.

\begin{lemma}
	\label{lem:secondterm} For any $c>0$,
	\[\limsup_{N\to\infty}\sup_{\bfk\in B_N}N^{1-c}\E_{\bfk}\pa{\int_{s=0}^{\bfH^{\sigma}(\partial E_{2N})}\ind{\bfX^{\sigma}(s)\in D_N}ds}=0.\]  
\end{lemma}

\begin{proof}[Proof of Lemma \ref{lem:secondterm}]
	In order to simplify the problem, we introduce a discrete time random walk $(\bfZ_m)_{m\geq 0}$ on $\Z^2$, started from $\bfk$ as well, and performing the exact same jumps as $\bfX^{\sigma}$. Then, shortening $\bfH^\bfZ:=\bfH^\bfZ(\partial E_{2N})$ the (discrete) time at which $\bfZ$ reaches the boundary $\partial E_{2N}$, and since the waiting time of $\bfX^{\sigma}$ at any site in $D_N\subset \bar D$ has distribution $Exp(2N^2)$ (4 neighbors, each jumped to at rate $N^2/2$), we can write
	\begin{equation}
	\label{eq:RWYZ}
	\E_{\bfk}\pa{\int_{s=0}^{\bfH^\sigma(\partial E_{2N})}\ind{\bfX^{\sigma}(s)\in D_N}ds}=\frac{1}{2N^2}\E_{\bfk}\pa{\sum_{m=1}^{\bfH^\bfZ}\ind{\bfZ_m\in D_N}}.
	\end{equation}
	Recall that $E_{2N}=\{\bfk\in \Z^2,\quad \norm{\bfk}<2N\}$ is the discrete box of size $2N$. Fix some $\bfk, \bfk'\in E_{2N}$, we now compute 
	\[\psi_{\bfk}(\bfk'):=\E_{\bfk}\pa{\sum_{m=1}^{\bfH^\bfZ}\ind{\bfZ_m=\bfk'}}.\] 
	Since $\bfZ$ performs a symmetric random walk until reaching $\partial E_{2N}$, $\psi_{\bfk}$ is solution to 
	\[\begin{cases}
	(\boldsymbol\Delta_{N}\phi(\bfk')=0 & \forall \bfk'\in E_{2N}\setminus \{\bfk\}\\
	\phi(\bfk')=0&\forall \bfk'\in \partial E_{2N}\\
	\phi(\bfk)=1/p_{\bfk,N}&
	\end{cases},\]
	where $p_{\bfk,N}$ is the probability for $\bfZ$, starting from $\bfk$ to reach $\partial E_{2N}$ without coming back to $\bfk$, which is also, starting from $\bfk$, the expectation of the number of passages in $\bfk$ before reaching the boundary $\partial E_{2N}$. 
	By maximum principle, we can now crudely bound $\psi_{\bfk}(\bfk')$, uniformly in $\bfk'$, by $1/p_{\bfk,N}$, so that the right hand side in \eqref{eq:RWYZ} is bounded from above for any $\bfk$ by
	\[\frac{1}{2N^2}\sum_{\bfk'\in D_N}\psi_\bfk(\bfk')\leq \frac{\# D_N}{2N^2 p_{\bfk,N}}\leq\frac{1}{2N p_{\0,N}}.\]
	The last holds due to the probability to leave $\bfk$ and never come back before reaching the boundary $\partial E_{2N}$ being smallest for $\bfk=\0$, and because the cardinal of $D_N$ is $N-2\leq N$. As $N$ goes to infinity, we have $p_{\0,N}\geq K/\log N,$ so that for any $N$ large enough,
	\[\sup_{\bfk\in B_N}N^{1-c}\E_{\bfk}\pa{\int_{s=0}^{\bfH^\sigma(\partial E_{2N})}\ind{\bfX^{\sigma}(s)\in D_N}ds}\leq \frac{\log N}{2KN^{c}},\]
	for some fixed constant $K$. This concludes the proof of the Lemma.
\end{proof}

\subsection{Proof of Lemma \ref{lem:thirdterm}} Recall that $\varepsilon$ and $\delta$ are fixed, small, positive constants, that
\begin{equation}
\label{eq:DefDNtilde}
\widetilde{D}_{N,\varepsilon}=\{(k,k+1)\in D_N, \quad k\leq N^{1-\varepsilon/4}\mbox{ or }k\geq N-N^{1-\varepsilon/4}\}.
\end{equation}
and that we want to prove that $\Prob_{\bfk}\Big[\bfH(\partial B_N)>\bfH(\widetilde{D}_{N, \varepsilon})\Big]$  vanishes, as $N\to\infty$, uniformly in $t\in [0,T]$ and $\bfk\in B_N$ such that $\norm{\bfk}$, $\norm{\bfN-\bfk}>\delta N$.

Once again, let us make our problem symmetric w.r.t the diagonal $D_N$, and recalling the notations introduced after Lemma \ref{lem:thirdterm}, write
\[\Prob_\bfk\Big[\bfH(\partial B_N)>\bfH(\widetilde{D}_{N, \varepsilon})\Big]=\Prob_\bfk\Big[\bfH^\sigma(\partial B_N^\sigma)>\bfH^\sigma(\widetilde{D}_{N, \varepsilon})\Big].\]
(Of course, since $\widetilde{D}_{N, \varepsilon}\subset \bar D$, we have $\widetilde{D}_{N, \varepsilon}^\sigma=\widetilde{D}_{N, \varepsilon}$) For any $\bfk, \ell\in \Z^2\times \N$, let us denote 
\[E_\ell(\bfk)= \{\bfk'\in \Z^2, \norm{\bfk-\bfk'}\leq \ell\},\]
\[\partial E_\ell(\bfk)= \{\bfk'\in \Z^2, \norm{\bfk-\bfk'}= \ell\}.\]
Then, letting $\ell_N=2N^{1-\varepsilon/4}$,  we have \[\widetilde{D}_{N,\varepsilon}\subset E_{\ell_N}(\0)\cup E_{\ell_N}(\bfN),\]
	therefore by union bound,
\begin{multline}
\label{eq:unionbound}
\Prob_{\bfk}\Big[\bfH(\partial B_N)>\bfH(\widetilde{D}_{N, \varepsilon})\Big]\\
\leq\Prob_\bfk\Big[\bfH^\sigma(E_{\ell_N}(\0))<\bfH^\sigma(\partial B_N^\sigma)\Big]+\Prob_\bfk\Big[\bfH^\sigma(E_{\ell_N}(\bfN))<\bfH^\sigma(\partial B_N^\sigma)\Big].
\end{multline}
Since we assume both $\norm{\bfk}$ and $\norm{\bfN-\bfk}$ to be larger than $\delta N$, both of the probabilities on the right hand side are estimated in the same way, so that we will only estimate the first one. To do so, simply note that for any $\bfk\in B_N$, $\bfH^\sigma(\partial B_N^\sigma)\leq \bfH^\sigma(\partial E_{2N})$, so that
\[\Prob_\bfk\Big[\bfH^\sigma(E_{\ell_N}(\0))<\bfH^\sigma(\partial B_N^\sigma)\Big]\leq \Prob_\bfk\Big[\bfH^\sigma(E_{\ell_N}(\0))<\bfH^\sigma(\partial E_{2N}^\sigma(\0))\Big].\]
The left hand side above can be written as \[\frac{\log(\norm{\bfk})-\log(2N)}{\log(\ell_N)-\log(2N)}+o_N(1),\]
where the $o_N(1)$ vanishes uniformly in $\bfk$. In particular, 
\[\sup_{\substack{\bfk\in B_N\\ 
\norm{\bfk}>\delta N}}\Prob_{\bfk}\Big[\bfH^\sigma(E_{\ell_N}(\0))<\bfH^\sigma(\partial B_N^\sigma)\Big]\leq \frac{\varepsilon}{4}\pa{\frac{ \log 2 -\log\delta }{\log N}}.\]
We obtain similarly
\[\sup_{\substack{\bfk\in B_N\\ 
\norm{\bfN-\bfk}>\delta N}}\Prob_{\bfk}\Big[\bfH^\sigma(E_{\ell_N}(\bfN))<\bfH^\sigma(\partial B_N^\sigma)\Big]\leq \frac{\varepsilon}{4}\pa{\frac{ \log 2 -\log\delta }{\log N}}.\]
Together with \eqref{eq:unionbound}, these two bounds conclude the proof of Lemma \ref{lem:thirdterm}.

\section{Proof of Theorem \ref{thm:HydroGen}} 
\label{sec:Hydro}
We now have all the tools needed to prove the hydrodynamic limit. Fix a continuous function $G: [0,1]\to \bb R$, and $t\in [0,T]$. Then, Using triangular and Cauchy Schwarz inequalities, we can estimate the square of the quantity inside the expectation in Theorem \ref{thm:HydroGen} by
\begin{multline*}
\Bigg( \frac 1N
\sum_{k\in \Lambda_N} G(k/N) \, [\eta_k(t) - \bar \rho(t,k/N)] \Bigg)^2\\
\leq
\frac{C(p)\norm{G}_{\infty}}{N}+\frac{2}{N^2}\pa{
\sum_{k=p+2}^{N-1} G(k/N) \big[\eta_k(t) -  \rho_N(t,k)\big] }^2\\
+\frac{2}{N^2}\pa{\sum_{k=p+2}^{N-1} G(k/N) \, \big[\rho_N(t,k) - \bar \rho(t,k/N)\big] }^2\\ 
\leq \frac{C(p)\norm{G}_{\infty}}{N}+
 \frac 2{N^2}
\sum_{k,l=p+2}^{N-1} G(k/N)G(l/N) \varphi_N(t,k,l)\\
+2\norm{G}_{\infty}^2 \frac 1N
\sum_{k=p+2}^{N-1} \big[\rho_N(t,k) - \bar \rho(t,k/N)\big]^2. 
\end{multline*}
For any positive $\delta$, the first sum on the right-hand side is less than 
\begin{multline*}
\sum_{\substack{\bfk\in B_N\\
\norm{\bfk}, \norm{\bfN-\bfk}>\delta N}}\frac{2\norm{G}_{\infty}^2\abs{\varphi_N(t,\bfk)}}{N^2}\\
+\frac{2\norm{G}_{\infty}^2}{N^2}\#\Big\{\bfk\in \{p+1,...,N-1\}^2\quad \norm{\bfk}\wedge\norm{\bfN-\bfk}\leq \delta N\Big\}.\end{multline*}
The first term vanishes as $N\to \infty$ for any $\delta>0$ according to Proposition \ref{prop:BoundCor}, whereas the second converges as $N\to\infty$ to $C\delta^2$ for some constant $C$. We then let $\delta\to0$, so that Theorem \ref{thm:HydroGen} follows from Lemma  \ref{lem:Hydro2} below.

\begin{lemma}
	\label{lem:Hydro2}
	For any $t\in [0,T]$
	\[\limsup_{N\to\infty}\frac 1N
	\sum_{k=p+2}^{N-1} (\rho_N(t,k) - \bar \rho(t,k/N))^2 =0.\]  
\end{lemma}
\begin{proof}[Proof of Lemma \ref{lem:Hydro2}]To estimate the quantity above, we compute its time derivative
\begin{multline*}\partial_t\frac 1N \sum_{k=p+2}^{N-1}\pa{\rho_N(t,k) - \bar \rho(t,k/N)}^2\\
=\frac 2{N} \sum_{k=p+2}^{N-1}(\rho_N(t,k)-\bar \rho(t,k/N))(N^2(\Delta_N\rho_N)(t,k) - (\Delta\bar \rho)(t,k/N)).\end{multline*}
Note that because the boundary conditions are not a priori respected by the initial profile $\rho_0$, the space derivative of $\bar{\rho}$ can diverge as $t\to 0$. However, Since $\bar{\rho}$ is smooth, for any $\varepsilon>0$, uniformly in $t\in[\varepsilon, T]$, we can write 
\[(\Delta\bar \rho)(t,k/N)=N^2(\bar{\rho}(t,(k+1)/N)+\bar{\rho}(t,(k-1)/N)-2\bar{\rho}(t,k/N))+C(\varepsilon)o_N(1),\]
where $C(\varepsilon)$ can diverge as $\varepsilon\to 0$.
For any $k\in \llbracket-1,N\rrbracket$, let us denote 
\[\theta(t,k)=\rho_N(t,k)-\bar \rho (t, k/N).\]
Thanks to the identity above, for any $t\in [\varepsilon,T]$ we obtain by  integration by parts
\begin{multline*}
\partial_t\frac 1N \sum_{k=p+2}^{N-1}\pa{\rho_N(t,k) - \bar \rho(t,k/N)}^2=2N \sum_{k=p+2}^{N-1}\theta(t,k)(\Delta_N\theta)(t,k)+C(\varepsilon)o_N(1)\\
=-2N\Bigg[ \sum_{k=p+1}^{N}\big\{\theta(t,k+1)-\theta(t,k)\big\}^2+\theta(t,N)\big\{\theta(t,N)-\theta(t,N-1)\big\}\\
-\theta(t,p+1)\big\{\theta(t,p+2)-\theta(t,p+1)\big\}\Bigg]+C(\varepsilon)o_N(1) 
\end{multline*}

The first sum above is negative, and does therefore not need to be controlled. Furthermore, for any $\varepsilon>0$, the last term vanishes as $N\to\infty$.  We now take a look at the two other terms. They are treated in the same way, so that we only consider the second. The general idea is that $\theta$ is at most $O(N^{-\frac 12})$, whereas the second factor is a gradient of order at least $o(N^{-\frac 12-\delta})$. More precisely, 
\begin{align*}
N^{\frac 12}\abs{\theta(t,p+1)}=&N^{\frac 12}\abs{\rho_N(t,p+1)-\bar \rho (t, (p+1)/N)}\\
\leq& N^{\frac 12}\abs{\rho_N(t,p+1)-\alpha}+N^{\frac 12}\abs{\alpha-\bar \rho (t, (p+1)/N)}.
\end{align*}
For any $\varepsilon>0$, the first term vanishes uniformly in $t\in [\varepsilon , T]$ according to Proposition \ref{lem:convrhop}. The second term vanishes as well, uniformly in $t\in [\varepsilon,T]$, because $\bar{\rho}$ is smooth. Similarly, for any $\delta>0$, 
\begin{align*}
N^{\frac 12+\delta}\abs{\theta(t,p+2)-\theta(t,p+1)}\leq& N^{\frac 12+\delta}\abs{\rho_N(t,p+2)-\rho_N(t,p+1)}\\
&+N^{\frac 12+\delta}\abs{\bar \rho (t, (p+2)/N)-\bar \rho (t, (p+1)/N)}.
\end{align*}
Once again, both terms vanish uniformly in $t\in [\varepsilon , T]$ according to Proposition \ref{lem:convrhop} and because $\bar{\rho}$ is smooth. The term $N\theta(t,N)(\theta(t,N)-\theta(t,N-1))$ is estimated in the same fashion.

Finally, for any $\varepsilon$, and any $t\in [\varepsilon,T]$,
\begin{multline*}
\partial_t\frac 1N \sum_{k=p+2}^{N-1}\pa{\rho_N(t,k) - \bar \rho(t,k/N)}^2\\
\leq -2N \sum_{k=p+1}^{N}\pa{\theta(t,k+1)-\theta(t,k)}^2+C(\varepsilon)o_N(1).
\end{multline*}
We can thus write, for any $\varepsilon>0$ and $t\in [\varepsilon,T]$
\[\frac 1N \sum_{k=p+2}^{N-1}\pa{\rho_N(t,k) - \bar \rho(t,k/N)}^2\leq \frac 1N \sum_{k=p+2}^{N-1}\pa{\rho_N(\varepsilon,k) - \bar \rho(\varepsilon,k/N)}^2 + C(\varepsilon)o_N(1).\]
Lemma \ref{lem:Hydro2} therefore follows from Lemma \ref{lem:Tpsepsilon} below.
\end{proof}

\begin{lemma} 
\label{lem:Tpsepsilon}
\begin{equation*}
\limsup_{\varepsilon\to 0}\limsup_{N\to\infty}\frac 1N \sum_{k=p+2}^{N-1}\pa{\rho_N(\varepsilon,k) - \bar \rho(\varepsilon,k/N)}^2=0.
\end{equation*} 
\end{lemma}
\begin{proof}[Proof of Lemma \ref{lem:Tpsepsilon}]
For any $k\in \{p+1,N-1\}$ (resp. $u\in [0,1]$) let $\Prob_k$ (resp. $\widetilde{\Prob}_u$) be the distribution of a continuous time random walk $X$ on $\Z$ (resp. a standard Brownian motion $B$) started from $k$ (resp. from $u$), and jumping at rate $N^2$ to any of its neighbors. Let $\E_k$ (resp. $\widetilde{\E}_u$) denote the corresponding expectation. Fix $k\in \{p+2,N-1\}$, we write for any $\varepsilon $
\begin{equation}
 \label{eq:rho0X}
\rho_N(\varepsilon,k)=\E_k\Big[\rho_N(\varepsilon-H_\varepsilon,X(H_\varepsilon))\Big]
 \end{equation}
where we shortened $H_\varepsilon:=H(\{p+1, N-1\})\wedge\varepsilon$ and $H(\{p+1, N-1\})$ is  $X$'s hitting time of the boundary $\{p+1,N-1\}$. Similarly, 
\begin{equation}
 \label{eq:rho0B}
\bar \rho(\varepsilon,k/N)=\widetilde{\E}_{k/N}\Big[\rho_N(\varepsilon-\widetilde{H}_{\varepsilon},B(\widetilde{H}_{\varepsilon}))\Big],
 \end{equation}
where $\widetilde{H}_{\varepsilon}=\widetilde{H}(\{0,1\})\wedge \varepsilon$ and $\widetilde{H}(\{0,1\})$ is $B$'s hitting time of the boundary $\{0,1\}$.

\medskip

We now use both of these identities to prove that, at distance at least $ \varepsilon^{1/4}N$ from the boundary, the density is close to its initial value.
Fix $k$ such that 
\[p+2+\varepsilon^{1/4}N<k<N-1-\varepsilon^{1/4}N.\] 
Then, we obtain from \eqref{eq:rho0X}, since $\rho_N$ is bounded in absolute value by $1$,
\begin{multline*}
\abs{\rho_N(\varepsilon,k)-\rho_0(k/N)}\leq 2 \Prob_k\pa{\abs{X(H_\varepsilon)-k}\geq \varepsilon^{1/4}N}\\
+\sup_{\abs{k-k'}\leq \varepsilon^{1/4}N}\abs{\rho_0(k/N)-\rho_0(k'/N)}.
 \end{multline*}
In a time $\varepsilon$, $X$ would typically travel a distance $\sqrt{\varepsilon} N$, so that the first term is $O(e^{-\varepsilon^{-1/4}})$. Since $\rho_0$ is smooth, the second term is $O(\varepsilon^{1/4})$. Using this time equation \eqref{eq:rho0B}, we can write an analogous bound for $\bar \rho$ so that for any $\varepsilon^{1/4}\leq u\leq 1-\varepsilon^{1/4}$
\begin{multline*}
\abs{\bar \rho(\varepsilon,u)-\rho_0(u)}\leq 2 \widetilde{\Prob}_u\pa{|B(\widetilde{H}_{\varepsilon})-u|\geq \varepsilon^{1/4}}+\sup_{\abs{u-u'}\leq \varepsilon^{1/4}}\abs{\rho_0(u)-\rho_0(u')}\\
=O_{\varepsilon}(e^{-\varepsilon^{-1/4}})+O_{\varepsilon}(\varepsilon^{1/4}).
 \end{multline*}
We finally obtain for any $p+2+\varepsilon^{1/4}N<k<N-1-\varepsilon^{1/4}N$
\[\abs{\rho_N(\varepsilon,k)-\bar\rho(\varepsilon, k/N)}=o_{\varepsilon}(\varepsilon),\]
where the $o_{\varepsilon}(1)$ is uniform in $k$ and can be chosen independent of $N$. 

\medskip

Since for any $k$,
\[\abs{\rho_N(\varepsilon,k) - \bar \rho(\varepsilon,k/N)}\leq1,\]
we can now estimate 
\begin{multline*}\frac 1N \sum_{k=p+2}^{N-1}\pa{\rho_N(\varepsilon,k) - \bar \rho(\varepsilon,k/N)}^2\\
\leq \frac 1N \sum_{k=p+2+\varepsilon^{1/4}N}^{N-1-\varepsilon^{1/4}N}\pa{\rho_N(\varepsilon,k) - \bar \rho(\varepsilon,k/N)}^2+\frac{p+2}{N}+2\varepsilon^{1/4}=O_{\varepsilon}(1)+O_{N}(1),\end{multline*}
where $O_{\varepsilon}(1)$ does not depend on $N$, which proves Lemma \ref{lem:Tpsepsilon}.
\end{proof}

\end{document}